\documentclass[%
 reprint,amsmath,amssymb,aps,
]{revtex4-2}
\usepackage[utf8]{inputenc}
\usepackage{graphicx}
\usepackage{dcolumn}
\usepackage{bbm}
\usepackage{amssymb}
\usepackage{MnSymbol}
\usepackage{amsfonts}
\usepackage{mathtools}
\usepackage{amsmath}
\usepackage{amsthm}
\usepackage{xcolor}
\usepackage{hyperref}
\usepackage{comment}
\usepackage{physics}
\usepackage{wasysym}
\usepackage{multirow}
\usepackage{booktabs}
\usepackage{tabularx}
\usepackage{array}

\usepackage{xcolor}

\definecolor{cA}{HTML}{D62728}
\definecolor{cB}{HTML}{1F77B4}
\definecolor{cC}{HTML}{2CA02C}
\definecolor{cD}{HTML}{FF7F0E}
\definecolor{cE}{HTML}{9467BD}
\definecolor{cF}{HTML}{8C564B}
\definecolor{cG}{HTML}{E377C2}
\definecolor{cH}{HTML}{17BECF}
\definecolor{cI}{HTML}{7F7F00}
\definecolor{cJ}{HTML}{7F7F7F}

\usepackage{tikz}
\usetikzlibrary{matrix,positioning}

\newcommand\mc[1]{\mathcal{#1}}

\newcommand{\cB}{{\mc{B}}}
\newcommand{\cC}
{{\mc{C}}}

\newcommand{\cF}{{\mc{F}}}
\newcommand{\cH}{{\mc{H}}}

\newcommand{\cK}{{\mc{K}}}
\newcommand{\cL}{{\mc{L}}}
\newcommand{\cM}{{\mc{M}}}
\newcommand{\cN}{{\mc{N}}}

\newcommand{\cP}{{\mc{P}}}
\newcommand{\cQ}{{\mc{Q}}}
\newcommand{\cR}{{\mc{R}}}

\newcommand{\cS}{{\mc{S}}}
\newcommand{\cT}{{\mc{T}}}

\newcommand{\cW}{{\mc{W}}}

\newcommand{\fq}{{\mathfrak{q}}}

\newcommand{\cU}{{\mc{U}}}
\newcommand{\cV}{{\mc{V}}}

\newcommand{\R}{{\mathbb R}}

\newcommand{\Gr}{{\mathbf{Gr}}}

\newcommand{\one}{{\mathbbm{1}}}

\newcommand{\zz}{{\mathbb{Z}}}
\newcommand{\ra}{{\rightarrow}}

\newcommand{\Sp}{{\mathrm{Sp}}}

\newcommand{\Stab}{{\mathrm{Stab}}}

\newcommand{\Lag}{{\mathrm{Lag}}}

\newcommand{\PG}{{\mathrm{PG}}}

\newcommand{\bracket}[2]{\left\langle#1\,\middle|\,#2\right\rangle}

\newtheorem{definition}{Definition}
\newtheorem{lemma}{Lemma}
\newtheorem{corollary}{Corollary}
\newtheorem{proposition}{Proposition}
\newtheorem{theorem}{Theorem}

\begin{document}
\preprint{APS/123-QED}

\title{Unextendible stabiliser bases}
\author{Markus Frembs}
\email{markus.frembs@itp.uni-hannover.de}
\affiliation{Institut f\"ur Theoretische Physik, Leibniz Universit\"at Hannover, Appelstraße 2, 30167 Hannover, Germany}

\begin{abstract}
    We study incomplete sets of orthogonal stabiliser states that cannot be extended by a further stabiliser state orthogonal to all of its members. Such sets are the stabiliser analogue of unextendible product bases (UPBs), and we thus call them unextendible stabiliser bases (USBs). Leveraging the symplectic geometry underlying the $n$-qudit Pauli group in prime local dimension, we explicitly construct USBs for systems of four qubits and three qudits of odd prime local dimension. Moreover, we show that these are the respective minimal qubit, respectively qudit numbers for which such bases exist, and that USBs exist for all $n\geq4$ qubit and for all $n\geq3$ odd-prime-dimensional qudit systems. Finally, we compare the resource-theoretic aspects of USBs with those of UPBs. We establish that, analogous to the case of UPBs, the orthogonal complement of every unextendible stabiliser set is a stabiliser-free subspace, and its normalised projector is necessarily magic; moreover, it is bound magic in odd prime dimension, yet need not be for qubits. We also show that USB unextendibility alone imposes no uniform quantitative obstruction to discrimination by stabiliser operations.
\end{abstract}

\maketitle

\section{Introduction}

The stabiliser formalism is a distinguished, efficiently classically simulable subtheory of finite-dimensional quantum mechanics. Pure stabiliser states are prepared from Pauli eigenstates using Clifford transformations, while stabiliser operations additionally allow Pauli measurements, classical control, stabiliser ancillas and discarding of subsystems \cite{Gottesman1998,AaronsonGottesman2004,DehaeneDeMoor2003,HostensDehaeneDeMoor2005,HeimendahlHeinrichGross2022}. Although these operations are not universal, supplementing them with suitable nonstabiliser states---or \emph{magic states}---is sufficient for universal quantum computation \cite{BravyiKitaev2005,Reichardt2005,BravyiHaah2012,Veitch2012}. This makes the geometry of stabiliser states, and the restrictions imposed on their preparation, transformation and measurement, central questions in the resource theory of magic \cite{HowardCampbell2017,VeitchEtAl2014}.

In this work, we study a basic completion problem for pairwise orthogonal stabiliser states. We call an orthonormal basis whose members are pure stabiliser states a \emph{stabiliser basis}, and ask whether every incomplete pairwise orthogonal family of pure stabiliser states can be enlarged to a stabiliser basis. The answer is negative: an incomplete orthogonal family may admit no additional pure stabiliser state orthogonal to all of its members. In analogy with unextendible product bases (UPBs), we call such a family an \emph{unextendible stabiliser basis} (USB).

Our main contribution consists of a general existence result of USBs for $n$-qudit systems of prime local dimension. More precisely, we construct examples of unextendible sets of orthogonal stabiliser states for 4-qubit and 3-qudit systems of odd prime local dimension, prove that these are the respectively minimal parties for which the phenomenon occurs, and further provide a lifting argument proving existence of USBs for all $n\geq4$ qubits and all $n\geq3$ qudits of odd prime local dimension.

In addition, we examine how far the analogy between USBs and UPBs extends operationally.

An unextendible product basis (UPB) is an incomplete orthogonal family of product states whose orthogonal complement contains no product vector \cite{BennettEtAl1999,DiVincenzoEtAl2003}. The normalised projector onto the complementary subspace of a UPB is bound entangled \cite{BennettEtAl1999}, while strong uncompletability under local extensions implies that the members of a UPB cannot be perfectly distinguished by LOCC \cite{DiVincenzoEtAl2003}. A USB presents a formally analogous obstruction internal to stabiliser theory: its orthogonal complement contains no pure stabiliser state. Consequently, every density operator supported entirely in that complementary subspace lies outside the stabiliser polytope and is therefore magic. For odd prime local dimension, we show that every USB complementary state is Wigner-positive and thus also bound magic. This conclusion fails for qubits: for our four-qubit construction below we give an explicit stabiliser protocol producing a state in a known magic-distillable region \cite{BravyiKitaev2005,CampbellBrowne2010,VeitchEtAl2014}. Finally, unextendibility implies no uniform obstruction to state discrimination. This should be contrasted with recent examples of pairwise orthogonal stabiliser states that cannot be perfectly distinguished by stabiliser operations \cite{Kwon2025}. The distinction shows that geometric unextendibility, stabiliser uncompletability and operational indistinguishability are related but inequivalent notions.

\section{Symplectic formalism}

We revisit the orthogonality relations for stabiliser states in prime local dimension \cite{Gottesman1998,Gottesman1999,HowardBrennanVala2015}, and define the notion of an unextendible stabiliser basis (USB).\\

\textbf{Orthogonality of stabiliser states.} Let $q$ be prime, $D=q^n$, and equip $V=\mathbb F_q^{2n}$ with its standard nondegenerate symplectic form $\omega$. A subspace $U\subset V$ is \emph{isotropic} if $\omega|_{U\times U}=0$, and \emph{Lagrangian} if it is maximal isotropic. We denote the set of Lagrangian subspaces of $V$ by $\Lag(V)$.

Fix a Weyl section $V\ni v\mapsto W_v\in\cP_{n,q}$, where $\cP_{n,q}$ denotes the $n$-qudit Pauli group, such that
\begin{align*}
    \tr[W_u^\dagger W_v]
    =D\delta_{u,v}\; ,
    \qquad
    \tr\left(W_uW_v\right)
    =D\delta_{u,-v}\; .
\end{align*}
For odd $q$, we may choose the standard phase convention
\begin{align*}
    W_uW_v
    =\zeta_q^{-\omega(u,v)/2}W_{u+v}\; ,
\end{align*}
for $\zeta_q=e^{2\pi i/q}$ and $1/2$ the inverse of $2$ in $\mathbb F_q$. In particular, the Weyl operators multiply without an additional phase on every isotropic subspace. A pure stabiliser state can thus be labelled by a pair $(L,\xi)$ with $L\in\Lag(V)$ and $\xi\in L^*$, where $L^*$ is the space of linear functionals on $L$.

For qubits, this is not possible since the central extension $1\ra\zz_4\ra\cP_{n,2}\ra V\ra 1$ does not split, equivalently the cohomology class of the $2$-cocycle $\gamma_W$ corresponding to this extension is nontrivial. Instead, choose a Hermitian Weyl section such that for $v,u\in V$ with $\omega(u,v)=0$,
\begin{align*}
    W_uW_v
    =(-1)^{\beta_W(u,v)}W_{u+v}\; ,
\end{align*}
where $\beta_W=\gamma_W/2$ (for $\omega(u,v)=0$) is the corresponding $\mathbb F_2$-valued multiplication cocycle \cite{Raussendorf2019,RaussendorfEtAl2023}. The admissible eigenvalues over $L\in\Lag(V)$ thus form the affine space
\begin{align*}
    X_{\beta_W}(L)
    =\{\xi:L\longrightarrow\mathbb F_2\mid\delta\xi=\beta_W|_{L\times L}\}\; ,
\end{align*}
where $(\delta\xi)(u,v)=\xi(u)+\xi(v)+\xi(u+v)$. 

To treat both cases (even/odd prime) uniformly, set
\begin{align*}
    \Xi(L)=
    \begin{cases}
        X_{\beta_W}(L) & q=2\\
        L^* & q\text{ odd}
    \end{cases}\; ,
\end{align*}
and for $\xi\in\Xi(L)$, define the eigenvalue function
\begin{align*}
    \chi_\xi(v)=
    \begin{cases}
        (-1)^{\xi(v)} & q=2\\
        \zeta_q^{\xi(v)} & q\text{ odd}
    \end{cases}\; .
\end{align*}
The corresponding pure stabiliser state $|L,\xi\rangle$ is the unique joint eigenstate satisfying $W_v|L,\xi\rangle=\chi_\xi(v)|L,\xi\rangle$ for every $v\in L$, and its rank-one projector is
\begin{align}\label{eq: stabiliser projector}
    \Pi_{L,\xi}
    =\dyad{L,\xi}
    =\frac{1}{D}\sum_{v\in L}\overline{\chi_\xi(v)}W_v\; .
\end{align}
We will write $\Stab^\mathrm{pure}_{n,q}$ for the set of pure, and $\Stab_{n,q}$ for the convex hull of all $n$-qudit stabiliser states.

\begin{lemma}[\cite{AaronsonGottesman2004,HostensDehaeneDeMoor2005}]\label{lm: stabiliser overlap}
    Let $|L,\xi\rangle$, $|M,\eta\rangle$ be two pure $n$-qudit stabiliser states. Then their squared overlap is given by
    \begin{align}\label{eq: stabiliser overlap}
        |\langle L,\xi|M,\eta\rangle|^2
        \ =\ \begin{cases}
            q^{\dim(L\cap M)-n} & \xi|_{L\cap M}=\eta|_{L\cap M}\\
            0 & \xi|_{L\cap M}\neq\eta|_{L\cap M}
        \end{cases}\; .
    \end{align}
\end{lemma}

\begin{proof}
    Let $R=L\cap M$. Using the Weyl expansions of the two projectors and the relation $\tr[W_uW_v]=D\delta_{u,-v}$ gives
    \begin{align*}
        |\langle L,\xi|M,\eta\rangle|^2
        &=\tr\left(\Pi_{L,\xi}\Pi_{M,\eta}\right)\\
        &=\frac{1}{D^2}\sum_{u\in L}\sum_{v\in M}\overline{\chi_\xi(u)}\overline{\chi_\eta(v)}\tr[W_uW_v]\\
        &=\frac{1}{D}\sum_{u\in R}\overline{\chi_\xi(u)}\chi_\eta(u)\; .
    \end{align*}
    Rephasing the Weyl section by $W_v\mapsto(-1)^{s(v)}W_v$ sends
    \begin{align*}
        \beta_W\longmapsto\beta_W+\delta s\; ,\qquad
        \xi\longmapsto\xi+s|_L\; ,
    \end{align*}
    and therefore preserves all restriction-equality conditions. In particular, note that orthogonality is independent of the chosen Weyl representative.
    
    For odd $q$, the summand is the additive character $u\mapsto\zeta_q^{\eta(u)-\xi(u)}$
    of $R$. For $q=2$, the restrictions of $\xi$ and $\eta$ have the same coboundary $\beta_W|_{R\times R}$, so $\eta|_R-\xi|_R$ is linear and $u\mapsto(-1)^{\eta(u)-\xi(u)}$ is again an ordinary character of $R$. Character orthogonality therefore gives (in both cases)
    \begin{equation*}
        \sum_{u\in R}\overline{\chi_\xi(u)}\chi_\eta(u)=
        \begin{cases}
            |R| & \xi|_R=\eta|_R\\
            0 & \xi|_R\neq\eta|_R
        \end{cases}\; .\qedhere
    \end{equation*}
\end{proof}

In particular, orthogonality requires
$L\cap M\neq\{0\}$, since two labels necessarily agree on the zero subspace.

For odd $q$, Eq.~(\ref{eq: stabiliser overlap}) admits a useful affine-geometric reformulation. To see this, note that the symplectic form induces an isomorphism $V/L\longrightarrow L^*$ given by
\begin{align*}
    a+L\longmapsto\xi_a\; ,\qquad
    \xi_a(v)=\omega(a,v)\; .
\end{align*}
Indeed, the kernel of the map $a\mapsto\xi_a$ is $L^\perp=L$, and both $V/L$ and $L^*$ have dimension $n$. Thus every pair $(L,\xi)$ is represented by a unique affine Lagrangian subspace $A=a+L$ such that $\xi=\xi_a$, and we write $|A\rangle:=|L,\xi_a\rangle$.

\begin{corollary}[\cite{Gross2006}]\label{cor: affine stabiliser overlap}
    Let $A=a+L$ and $B=b+M$ be affine Lagrangian subspaces of $V$, with $q$ odd. Then
    \begin{align}\label{eq: affine stabiliser overlap}
        |\langle A|B\rangle|^2
        \ =\ \frac{1}{D}|A\cap B|\; .
    \end{align}
\end{corollary}

\begin{proof}
    Note that the affine subspaces intersect if and only if $a-b\in L+M$. Since $L$ and $M$ are Lagrangian,
    \begin{align*}
        (L+M)^\perp
        =L^\perp\cap M^\perp
        =L\cap M
        =:R\; ,
    \end{align*}
    and thus $L+M=R^\perp$. It follows that
    \begin{align*}
        A\cap B\neq\emptyset
        &\Longleftrightarrow a-b\in R^\perp\\
        &\Longleftrightarrow \omega(a-b,v)=0\ \forall v\in R
        \Longleftrightarrow \xi_a|_R=\xi_b|_R\; .
    \end{align*}
    If the intersection is nonempty, it is an affine subspace with linear component $R$, and therefore $|A\cap B|=q^{\dim R}$, from which the claim follows with Lm.~\ref{lm: stabiliser overlap}.
\end{proof}

In particular, $\langle A|B\rangle=0$ if and only if $A\cap B=\emptyset$.\\

\textbf{Unextendible stabiliser bases.} Let $q$ be prime.

\begin{definition}\label{def: USB}
    Let $\cU=\left\{\Pi_{L_i,\xi_i}\right\}_{i=1}^k$, $k<D$ be a set of pairwise orthogonal $n$-qudit stabiliser states. Then $\cU$ is an \emph{unextendible stabiliser basis (USB)} if no other $n$-qudit stabiliser state is orthogonal to every element in $\cU$.
\end{definition}

Our terminology mimics that of unextendible product bases (UPBs), see also Sec.~\ref{sec: UPBs vs USBs}.

For odd $q$, Cor.~\ref{cor: affine stabiliser overlap} shows that a set of orthogonal stabiliser states corresponds to a set of pairwise disjoint affine Lagrangian subspaces. Such a set is unextendible precisely when every other affine Lagrangian intersects at least one member of the set nontrivially.

For qubits, the same obstruction is formulated in the cocycle-twisted fibres $X_{\beta_W}(L)$. A further stabiliser state $\Pi_{M,\eta}$ is orthogonal to every member of $\cU$ precisely when its affine character disagrees with each $\xi_i$ on $M\cap L_i$.

\section{USBs for qudits in prime dimension}\label{sec: unextendibility}

We provide examples of unextendible sets of pairwise orthogonal $n$-qudit stabiliser states (USBs) for $q$ even and odd prime. In Sec.~\ref{sec: unextendible sets of orthogonal stabiliser states}, we provide an explicit example for $n=4$ qubits, and in Sec.~\ref{sec: unextendible sets of orthogonal qudit stabiliser states} for $n=3$ qudits in odd prime dimension. The general existence result is Thm.~\ref{thm: unextendibility} in Sec.~\ref{sec: general USB existence}. We also treat the analogue assertion for the symplectic theory in the case $q=2$ in Sec.~\ref{app: unextendible sets of orthogonal symplectic stabiliser vectors}.

\subsection{A $4$-qubit USB}\label{sec: unextendible sets of orthogonal stabiliser states}

We will write $\{\one_i,X_i,Y_i,Z_i\}$ for the Pauli operators on the $i$-th qubit, that is, $X_i=\one\otimes X\otimes \one$. We will use the common shorthand and omit tensor products, both for Pauli strings, e.g. $X_i=\one X\one$ and stabiliser states, e.g. $\ket{z_+x_-y_+}=\ket{z_+}\otimes\ket{x_-}\otimes\ket{y_+}$ where $\ket{x_\pm},\ket{y_\pm},\ket{z_\pm}$ denote the $\pm1$-eigenstates of the respective Pauli operators.

\begin{theorem}\label{thm: 4 qubit USB}
    The following pairwise orthogonal $4$-qubit stabiliser states define an unextendible stabiliser basis:
    \begin{align*}
        \cS_4
        :=\bigl\{
         &\ket{z_+z_+z_+z_+},
         \ket{z_+z_+z_-x_+},
         \ket{z_-x_+x_+y_+},
         \ket{z_-x_+y_-y_-},\\
         &\ket{x_+z_-x_-y_+},
         \ket{x_+z_-y_+y_-},
         \ket{x_-x_-z_+z_-},
         \ket{x_-x_-z_-x_-}
        \bigr\}.
    \end{align*}
\end{theorem}

\begin{proof}
    The details can be found in App.~\ref{app: proof - 4 qubits}.
\end{proof}

Consequently, not every orthogonal set of stabiliser states is extendible to a stabiliser basis.

\subsection{A $3$-qudit USB (for $q$ odd prime)}\label{sec: unextendible sets of orthogonal qudit stabiliser states}

For $q$ odd prime, the single-qudit Pauli operators are
\begin{align*}
    X\ket j
    &=\ket{j+1}\; ,&
    Z\ket j
    &=\zeta^j_q\ket j\; ,&
    W&:=\zeta^{\frac{1}{2}}_qXZ\; .
\end{align*}
We denote their eigenbases by $\{\ket{x_a}\}_{a\in\mathbb F_q}$, $\{\ket{z_a}\}_{a\in\mathbb F_q}$, and $\{\ket{w_a}\}_{a\in\mathbb F_q}$, respectively. In the following, we again omit writing tensor products, and denote by $\mathbb F_q^\times\cong \zz_{q-1}$ the multiplicative units in $\mathbb F_q$.

\begin{theorem}\label{thm: 3 qudit USB}
    Let $q$ be odd prime. The following set of pairwise orthogonal three-qudit stabiliser states defines an unextendible stabiliser basis:
    \begin{equation}\label{eq: unextendible qudit stabiliser set}
    \begin{aligned}
        \cS_q
        :=\bigl\{\ket{x_az_bx_c}\mid a,b,c\in\mathbb F_q^\times\bigr\}
        \ &\cup\ 
         \bigl\{\ket{w_0z_0w_0}\bigr\}\\
        \cup\
        \bigl\{\ket{x_0w_bw_c}\mid b,c\in\mathbb F_q^\times\bigr\}
        \ &\cup\ 
         \bigl\{\ket{w_aw_0x_0}\mid a\in\mathbb F_q^\times\bigr\}\; .
    \end{aligned}
    \end{equation}
\end{theorem}

\begin{proof}
    The details are provided in App.~\ref{app: proof - 3 qudits}.
\end{proof}

\subsection{General existence results for USBs}\label{sec: general USB existence}

Our general existence result for USBs establishes the number of parties in the above examples to be minimal, and lifts them to larger ones.

\begin{theorem}\label{thm: unextendibility}
    Unextendible sets of orthogonal $n$-qudit stabiliser states with $q$ prime exist if and only if
    \begin{enumerate}
        \item $n\geq 4$ for qubits ($q=2$), or
        \item $n\geq 3$ for qudits ($q$ odd).
    \end{enumerate}
\end{theorem}

\begin{proof}
    Existence follows from Thm.~\ref{thm: 4 qubit USB} and Thm.~\ref{thm: 3 qudit USB}, and minimality from Thm.~\ref{thm: 2 qudit extendibility} and Thm.~\ref{thm: 3-qubit extendibility}, which establish that $2$-qudit and $3$-qubit sets of orthogonal stabiliser states always extend to a stabiliser basis. Every set of pairwise orthogonal single-qudit stabiliser states belongs to a single Pauli eigenbasis and can therefore be completed. The remaining cases follow by Lm.~\ref{lm: lifting} in App.~\ref{app: minimality and lifting}.
\end{proof}

\subsection{USBs for symplectic $n$-qubit stabiliser states}\label{app: unextendible sets of orthogonal symplectic stabiliser vectors}

The even and odd prime cases above are distinguished by the fact that the qubit Pauli extension is non-split. For completeness, we also treat the untwisted case (for $q=2)$, with `symplectic $n$-qubit stabiliser states' $|A\rangle=|a+L\rangle$ and orthogonality as in Cor.~\ref{cor: affine stabiliser overlap}. Unlike the qubit case (see Thm.~\ref{thm: unextendibility}), we find that unextendible sets of symplectic stabiliser states already exist for $n=3$.

\begin{theorem}\label{thm: 3 qubits symplectic}
    Let $(V=\zz^6_2,\omega)$ be the symplectic space of three qubits. Then there exists a set of five symplectic stabiliser states that cannot be extended to an orthogonal basis of symplectic stabiliser states.
\end{theorem}

\begin{proof}
    The details are provided in App.~\ref{app: proof - 3 qubits}.
\end{proof}

Again, here $n=3$ is minimal in the number of parties, and the result extends to all $n\geq 3$, as both Thm.~\ref{thm: 2 qudit extendibility} and Lm.~\ref{lm: lifting} are readily extended to the symplectic level.

\section{Relation with UPBs}\label{sec: UPBs vs USBs}

We compare Def.~\ref{def: USB} with the notion of unextendible product bases and their respective roles in the resource theories of local operations and classical communication (LOCC) and non-stabiliserness (or `magic').\\

\textbf{LOCC.} An unextendible product basis (UPB) is an incomplete set of pairwise orthogonal product states whose orthogonal complement contains no product state~\cite{BennettEtAl1999,DiVincenzoEtAl2003}. Let
$\cU=\{|\psi_i\rangle\}_{i=1}^k$
be a UPB, and define
\begin{align}\label{eq: unextendible subspace}
    P_\cU
    =\sum_{i=1}^k\dyad{\psi_i}\; ,\qquad
    Q_\cU
    =\one-P_\cU\; .
\end{align}
The normalised complementary projector $\rho_\cU^{\perp}=\frac{Q_\cU}{\tr[Q_\cU]}$ is entangled, since a separable decomposition would place a product vector in $\mathrm{ran}(Q_\cU)$. Moreover, partial transposition maps the pairwise orthogonal product states in $\cU$ to another pairwise orthogonal product set, consequently,
\begin{align*}
    \left(\rho_\cU^{\perp}\right)^{T_A}
    \geq0\; ,
\end{align*}
across any bipartition $(A,\overline{A})$. Since positive partial transpose implies non-distillability, $\rho_\cU^{\perp}$ is bound entangled \cite{BennettEtAl1999}. Moreover, UPBs satisfy an even stronger extension property: although unextendibility is destroyed by enlarging the local Hilbert spaces of a bipartition, a UPB cannot be completed to a full orthogonal product basis in any locally extended system. This property is called \emph{strong uncompletability} \cite{DiVincenzoEtAl2003}. Together with local Naimark dilation, it implies that the members of a UPB cannot be perfectly distinguished by LOCC~\cite{BennettEtAl1999,DiVincenzoEtAl2003}.\\

\textbf{Magic.} A closely analogous structure arises in the resource theory of magic. Let
$\cU=\{|\psi_i\rangle\}_{i=1}^k$ be an incomplete set of pairwise orthogonal $n$-qudit stabiliser states, and define $P_\cU$ and $Q_\cU$ as in Eq.~(\ref{eq: unextendible subspace}). If $\cU$ is a USB, then
\begin{align*}
    \mathrm{ran}(Q_\cU)\cap\Stab^\mathrm{pure}_{n,q}
    =\emptyset\; ,
\end{align*}
and its complementary subspace is then completely magic in the sense that every density operator supported in $\mathrm{ran}(Q_\cU)$ lies outside the stabiliser polytope. Indeed, if a density matrix $\rho=\sum_s p_s\dyad{s}$ with $p_s>0$ and $|s\rangle\in\Stab^\mathrm{pure}_{n,q}$ were supported in $\mathrm{ran}(Q_\cU)$, then
\begin{align*}
    0
    =\tr[P_\cU\rho]
    =\sum_s p_s\langle s|P_\cU|s\rangle\; .
\end{align*}
It follows that every stabiliser state occurring with $p_s>0$ would be orthogonal to every member of $\cU$, contradicting unextendibility. In turn, it is readily seen that a USB defines a magic witness (see Cor.~\ref{cor: USB magic witness} in App.~\ref{app: applications}), and there is also a direct stabiliser analogue of strong UPB uncompletability (see Lm.~\ref{lm: stabiliser uncompletability} App.~\ref{app: applications}).

Nevertheless, the analogy with UPBs (partly) breaks for two stronger operational conclusions. First, whether the complementary state $\rho_\cU^{\perp}=\frac{Q_\cU}{\tr[Q_\cU]}$ is bound magic depends on dimension. For UPB complements bound entanglement follows from positivity under partial transposition (PPT) \cite{Peres1996,Horodeckisz1998}, which is stable under tensor powers and thus obstructs entanglement distillation. For $q$ odd, an analogous role is played by positivity of the Wigner representation of stabiliser states \cite{Gross2006} (see Cor.~\ref{cor: bound magic states} in App.~\ref{app: applications}), which implies that $\rho_\cU^{\perp}$ is indeed bound magic. However, qubit stabiliser states do not admit a (canonical) Wigner representation and, indeed, the complementary state of the $4$-qubit USB in Thm.~\ref{thm: 4 qubit USB} is not bound magic (see Prop.~\ref{prop: distillability} in App.~\ref{app: applications}).

Second, stabiliser uncompletability does not exclude discrimination by coarse-grained stabiliser measurements, let alone by stabiliser operations which also include adaptive protocols with Pauli measurements, discarding and feed-forward \cite{VeitchEtAl2014}. Indeed, USB unextendibility imposes no uniform quantitative discrimination gap: there exist USBs whose uniform-prior discrimination success probability under stabiliser operations is arbitrarily close to one (see Prop.~\ref{prop: approximate SO distinguishability} in App.~\ref{app: applications}).

The above commonalities and differences are summarised in Tab.~\ref{tab: UPBs vs USBs}. Finally, we note that the two properties are not mutually exclusive, as the four-qubit example in Thm.~\ref{thm: 4 qubit USB} is both a USB and UPB (see also Ref.~\cite{Johnston2014}).\footnote{Write $S_4=\{\ket{u_i}\}_{i=1}^{8}$ in the order listed in Thm.~\ref{thm: 4 qubit USB}, and suppose $\ket{a_1}\!\ket{a_2}\!\ket{a_3}\!\ket{a_4}$ were orthogonal to every member of $S_4$. Assign to each $i$ a qubit $\varphi(i)$ on which the local overlap vanishes. Since a qubit state has a unique orthogonal state, all members of a block $\varphi^{-1}(j)$ must carry the \emph{same} local factor on qubit $j$. Grouping the elements of $S_4$ by equal local factor on qubit $j$ then gives:
\begin{align*}
 j=1,2:&\quad \{1,2\},\{3,4\},\{5,6\},\{7,8\}\\
 j=3:&\quad \{1,7\},\{2,8\},\{3\},\{4\},\{5\},\{6\}\\
 j=4:&\quad \{1\},\{2\},\{3,5\},\{4,6\},\{7\},\{8\}
\end{align*}
Note that every block has at most two elements, and since $|S_4|=8$ and there are four blocks, all four blocks must have exactly two elements. Since $\varphi^{-1}(3)\in\{\{1,7\},\{2,8\}\}$ and $\varphi^{-1}(4)\in\{\{3,5\},\{4,6\}\}$, the four resulting choices leave the residues $\{1,3,5,7\}$, $\{1,4,6,7\}$, $\{2,3,5,8\}$ and $\{2,4,6,8\}$, yet none of which contains an admissible block for qubit $1$ or $2$. Hence, no product vector is orthogonal to all elements of $S_4$, that is, $S_4$ is a UPB.}

\begin{table}[t!]
\caption{Comparison between unextendible product bases
(UPBs) and unextendible stabiliser bases (USBs).}
\label{tab: UPBs vs USBs}
\centering
\renewcommand{\arraystretch}{1.2}
\small
\begin{tabular}{c|c}
    \toprule
    \parbox[t]{0.46\columnwidth}{\centering UPB} & \parbox[t]{0.49\columnwidth}{\centering USB}\\
    \midrule
    \parbox[t]{0.46\columnwidth}{no product state in the orthogonal complement} & \parbox[t]{0.49\columnwidth}{no stabiliser state in the orthogonal complement}\\
    \addlinespace[3pt]
    \parbox[t]{0.46\columnwidth}{
    completely entangled complementary subspace} & \parbox[t]{0.49\columnwidth}{completely magic complementary subspace}\\
    \addlinespace[3pt]
    \parbox[t]{0.46\columnwidth}{
    strongly uncompletable under local extensions} & \parbox[t]{0.44\columnwidth}{strongly uncompletable under Clifford isometries}\\
    \addlinespace[3pt]
    \midrule
    \parbox[t]{0.46\columnwidth}{complementary state is bound entangled} & \parbox[t]{0.45\columnwidth}{complementary state is bound magic for $q$ odd, but need not be for $q=2$}\\
    \addlinespace[3pt]
    \midrule    
    \parbox[t]{0.46\columnwidth}{cannot be perfectly distinguished by LOCC} & \parbox[t]{0.49\columnwidth}{no uniform distinguishability gap by stabiliser operations}\\
    \bottomrule
\end{tabular}
\end{table}

\section{Conclusion}

Unextendible sets of pairwise orthogonal stabiliser states (USBs) expose a geometric obstruction internal to stabiliser theory: an incomplete orthogonal family can have a complementary subspace containing no further stabiliser state. Our main result establishes the existence of USBs for every $n\geq4$ qubit system and every $n\geq3$ odd-prime-dimensional qudit system. Moreover, we compared the operational properties of USBs within the resource theory of magic to those of unextendible product bases (UPBs) in the resource theory of LOCC. We showed that the normalised projector onto the complement of a USB is magic, and moreover bound magic in odd prime dimension, while this is not necessarily true for qubit USBs. We further showed that USB unextendibility alone imposes no uniform quantitative obstruction to discrimination by stabiliser operations.\\

\textbf{Acknowledgments.} The explicit constructions and proof ideas are due to OpenAI's GPT-5.5 and GPT-5.6 Sol models. I verified the arguments in detail, reformulated and refined some of the results and arranged them into the present manuscript.

\bibliography{bibliography}

\begin{thebibliography}{38}%
\makeatletter
\providecommand \@ifxundefined [1]{%
 \@ifx{#1\undefined}
}%
\providecommand \@ifnum [1]{%
 \ifnum #1\expandafter \@firstoftwo
 \else \expandafter \@secondoftwo
 \fi
}%
\providecommand \@ifx [1]{%
 \ifx #1\expandafter \@firstoftwo
 \else \expandafter \@secondoftwo
 \fi
}%
\providecommand \natexlab [1]{#1}%
\providecommand \enquote  [1]{``#1''}%
\providecommand \bibnamefont  [1]{#1}%
\providecommand \bibfnamefont [1]{#1}%
\providecommand \citenamefont [1]{#1}%
\providecommand \href@noop [0]{\@secondoftwo}%
\providecommand \href [0]{\begingroup \@sanitize@url \@href}%
\providecommand \@href[1]{\@@startlink{#1}\@@href}%
\providecommand \@@href[1]{\endgroup#1\@@endlink}%
\providecommand \@sanitize@url [0]{\catcode `\\12\catcode `\$12\catcode `\&12\catcode `\#12\catcode `\^12\catcode `\_12\catcode `\%12\relax}%
\providecommand \@@startlink[1]{}%
\providecommand \@@endlink[0]{}%
\providecommand \url  [0]{\begingroup\@sanitize@url \@url }%
\providecommand \@url [1]{\endgroup\@href {#1}{\urlprefix }}%
\providecommand \urlprefix  [0]{URL }%
\providecommand \Eprint [0]{\href }%
\providecommand \doibase [0]{https://doi.org/}%
\providecommand \selectlanguage [0]{\@gobble}%
\providecommand \bibinfo  [0]{\@secondoftwo}%
\providecommand \bibfield  [0]{\@secondoftwo}%
\providecommand \translation [1]{[#1]}%
\providecommand \BibitemOpen [0]{}%
\providecommand \bibitemStop [0]{}%
\providecommand \bibitemNoStop [0]{.\EOS\space}%
\providecommand \EOS [0]{\spacefactor3000\relax}%
\providecommand \BibitemShut  [1]{\csname bibitem#1\endcsname}%
\let\auto@bib@innerbib\@empty
\bibitem [{\citenamefont {Gottesman}(1998)}]{Gottesman1998}%
  \BibitemOpen
  \bibfield  {author} {\bibinfo {author} {\bibfnamefont {D.}~\bibnamefont {Gottesman}},\ }\bibfield  {title} {\bibinfo {title} {Theory of fault-tolerant quantum computation},\ }\href {https://doi.org/10.1103/PhysRevA.57.127} {\bibfield  {journal} {\bibinfo  {journal} {Phys. Rev. A}\ }\textbf {\bibinfo {volume} {57}},\ \bibinfo {pages} {127} (\bibinfo {year} {1998})}\BibitemShut {NoStop}%
\bibitem [{\citenamefont {Aaronson}\ and\ \citenamefont {Gottesman}(2004)}]{AaronsonGottesman2004}%
  \BibitemOpen
  \bibfield  {author} {\bibinfo {author} {\bibfnamefont {S.}~\bibnamefont {Aaronson}}\ and\ \bibinfo {author} {\bibfnamefont {D.}~\bibnamefont {Gottesman}},\ }\bibfield  {title} {\bibinfo {title} {Improved simulation of stabilizer circuits},\ }\href {https://doi.org/10.1103/PhysRevA.70.052328} {\bibfield  {journal} {\bibinfo  {journal} {Phys. Rev. A}\ }\textbf {\bibinfo {volume} {70}},\ \bibinfo {pages} {052328} (\bibinfo {year} {2004})}\BibitemShut {NoStop}%
\bibitem [{\citenamefont {Dehaene}\ and\ \citenamefont {De~Moor}(2003)}]{DehaeneDeMoor2003}%
  \BibitemOpen
  \bibfield  {author} {\bibinfo {author} {\bibfnamefont {J.}~\bibnamefont {Dehaene}}\ and\ \bibinfo {author} {\bibfnamefont {B.}~\bibnamefont {De~Moor}},\ }\bibfield  {title} {\bibinfo {title} {{Clifford group, stabilizer states, and linear and quadratic operations over GF(2)}},\ }\href {https://doi.org/10.1103/PhysRevA.68.042318} {\bibfield  {journal} {\bibinfo  {journal} {Phys. Rev. A}\ }\textbf {\bibinfo {volume} {68}},\ \bibinfo {pages} {042318} (\bibinfo {year} {2003})}\BibitemShut {NoStop}%
\bibitem [{\citenamefont {Hostens}\ \emph {et~al.}(2005)\citenamefont {Hostens}, \citenamefont {Dehaene},\ and\ \citenamefont {De~Moor}}]{HostensDehaeneDeMoor2005}%
  \BibitemOpen
  \bibfield  {author} {\bibinfo {author} {\bibfnamefont {E.}~\bibnamefont {Hostens}}, \bibinfo {author} {\bibfnamefont {J.}~\bibnamefont {Dehaene}},\ and\ \bibinfo {author} {\bibfnamefont {B.}~\bibnamefont {De~Moor}},\ }\bibfield  {title} {\bibinfo {title} {Stabilizer states and {C}lifford operations for systems of arbitrary dimensions and modular arithmetic},\ }\href {https://doi.org/10.1103/PhysRevA.71.042315} {\bibfield  {journal} {\bibinfo  {journal} {Phys. Rev. A}\ }\textbf {\bibinfo {volume} {71}},\ \bibinfo {pages} {042315} (\bibinfo {year} {2005})}\BibitemShut {NoStop}%
\bibitem [{\citenamefont {Heimendahl}\ \emph {et~al.}(2022)\citenamefont {Heimendahl}, \citenamefont {Heinrich},\ and\ \citenamefont {Gross}}]{HeimendahlHeinrichGross2022}%
  \BibitemOpen
  \bibfield  {author} {\bibinfo {author} {\bibfnamefont {A.}~\bibnamefont {Heimendahl}}, \bibinfo {author} {\bibfnamefont {M.}~\bibnamefont {Heinrich}},\ and\ \bibinfo {author} {\bibfnamefont {D.}~\bibnamefont {Gross}},\ }\bibfield  {title} {\bibinfo {title} {The axiomatic and the operational approaches to resource theories of magic do not coincide},\ }\href {https://doi.org/10.1063/5.0085774} {\bibfield  {journal} {\bibinfo  {journal} {Journal of Mathematical Physics}\ }\textbf {\bibinfo {volume} {63}},\ \bibinfo {pages} {112201} (\bibinfo {year} {2022})}\BibitemShut {NoStop}%
\bibitem [{\citenamefont {Bravyi}\ and\ \citenamefont {Kitaev}(2005)}]{BravyiKitaev2005}%
  \BibitemOpen
  \bibfield  {author} {\bibinfo {author} {\bibfnamefont {S.}~\bibnamefont {Bravyi}}\ and\ \bibinfo {author} {\bibfnamefont {A.}~\bibnamefont {Kitaev}},\ }\bibfield  {title} {\bibinfo {title} {Universal quantum computation with ideal clifford gates and noisy ancillas},\ }\href {https://doi.org/10.1103/PhysRevA.71.022316} {\bibfield  {journal} {\bibinfo  {journal} {Phys. Rev. A}\ }\textbf {\bibinfo {volume} {71}},\ \bibinfo {pages} {022316} (\bibinfo {year} {2005})}\BibitemShut {NoStop}%
\bibitem [{\citenamefont {Reichardt}(2005)}]{Reichardt2005}%
  \BibitemOpen
  \bibfield  {author} {\bibinfo {author} {\bibfnamefont {B.~W.}\ \bibnamefont {Reichardt}},\ }\bibfield  {title} {\bibinfo {title} {Quantum universality from magic states distillation applied to {CSS} codes},\ }\href {https://doi.org/10.1007/s11128-005-7654-8} {\bibfield  {journal} {\bibinfo  {journal} {Quantum Information Processing}\ }\textbf {\bibinfo {volume} {4}},\ \bibinfo {pages} {251} (\bibinfo {year} {2005})}\BibitemShut {NoStop}%
\bibitem [{\citenamefont {Bravyi}\ and\ \citenamefont {Haah}(2012)}]{BravyiHaah2012}%
  \BibitemOpen
  \bibfield  {author} {\bibinfo {author} {\bibfnamefont {S.}~\bibnamefont {Bravyi}}\ and\ \bibinfo {author} {\bibfnamefont {J.}~\bibnamefont {Haah}},\ }\bibfield  {title} {\bibinfo {title} {Magic-state distillation with low overhead},\ }\href {https://doi.org/10.1103/PhysRevA.86.052329} {\bibfield  {journal} {\bibinfo  {journal} {Phys. Rev. A}\ }\textbf {\bibinfo {volume} {86}},\ \bibinfo {pages} {052329} (\bibinfo {year} {2012})}\BibitemShut {NoStop}%
\bibitem [{\citenamefont {Veitch}\ \emph {et~al.}(2012)\citenamefont {Veitch}, \citenamefont {Ferrie}, \citenamefont {Gross},\ and\ \citenamefont {Emerson}}]{Veitch2012}%
  \BibitemOpen
  \bibfield  {author} {\bibinfo {author} {\bibfnamefont {V.}~\bibnamefont {Veitch}}, \bibinfo {author} {\bibfnamefont {C.}~\bibnamefont {Ferrie}}, \bibinfo {author} {\bibfnamefont {D.}~\bibnamefont {Gross}},\ and\ \bibinfo {author} {\bibfnamefont {J.}~\bibnamefont {Emerson}},\ }\bibfield  {title} {\bibinfo {title} {Negative quasi-probability as a resource for quantum computation},\ }\href@noop {} {\bibfield  {journal} {\bibinfo  {journal} {New J. Phys.}\ }\textbf {\bibinfo {volume} {14}},\ \bibinfo {pages} {113011} (\bibinfo {year} {2012})}\BibitemShut {NoStop}%
\bibitem [{\citenamefont {Howard}\ and\ \citenamefont {Campbell}(2017)}]{HowardCampbell2017}%
  \BibitemOpen
  \bibfield  {author} {\bibinfo {author} {\bibfnamefont {M.}~\bibnamefont {Howard}}\ and\ \bibinfo {author} {\bibfnamefont {E.}~\bibnamefont {Campbell}},\ }\bibfield  {title} {\bibinfo {title} {Application of a resource theory for magic states to fault-tolerant quantum computing},\ }\href {https://doi.org/10.1103/PhysRevLett.118.090501} {\bibfield  {journal} {\bibinfo  {journal} {Phys. Rev. Lett.}\ }\textbf {\bibinfo {volume} {118}},\ \bibinfo {pages} {090501} (\bibinfo {year} {2017})}\BibitemShut {NoStop}%
\bibitem [{\citenamefont {Veitch}\ \emph {et~al.}(2014)\citenamefont {Veitch}, \citenamefont {Hamed~Mousavian}, \citenamefont {Gottesman},\ and\ \citenamefont {Emerson}}]{VeitchEtAl2014}%
  \BibitemOpen
  \bibfield  {author} {\bibinfo {author} {\bibfnamefont {V.}~\bibnamefont {Veitch}}, \bibinfo {author} {\bibfnamefont {S.~A.}\ \bibnamefont {Hamed~Mousavian}}, \bibinfo {author} {\bibfnamefont {D.}~\bibnamefont {Gottesman}},\ and\ \bibinfo {author} {\bibfnamefont {J.}~\bibnamefont {Emerson}},\ }\bibfield  {title} {\bibinfo {title} {The resource theory of stabilizer quantum computation},\ }\href {https://doi.org/10.1088/1367-2630/16/1/013009} {\bibfield  {journal} {\bibinfo  {journal} {New Journal of Physics}\ }\textbf {\bibinfo {volume} {16}},\ \bibinfo {pages} {013009} (\bibinfo {year} {2014})}\BibitemShut {NoStop}%
\bibitem [{\citenamefont {Bennett}\ \emph {et~al.}(1999)\citenamefont {Bennett}, \citenamefont {DiVincenzo}, \citenamefont {Mor}, \citenamefont {Shor}, \citenamefont {Smolin},\ and\ \citenamefont {Terhal}}]{BennettEtAl1999}%
  \BibitemOpen
  \bibfield  {author} {\bibinfo {author} {\bibfnamefont {C.~H.}\ \bibnamefont {Bennett}}, \bibinfo {author} {\bibfnamefont {D.~P.}\ \bibnamefont {DiVincenzo}}, \bibinfo {author} {\bibfnamefont {T.}~\bibnamefont {Mor}}, \bibinfo {author} {\bibfnamefont {P.~W.}\ \bibnamefont {Shor}}, \bibinfo {author} {\bibfnamefont {J.~A.}\ \bibnamefont {Smolin}},\ and\ \bibinfo {author} {\bibfnamefont {B.~M.}\ \bibnamefont {Terhal}},\ }\bibfield  {title} {\bibinfo {title} {Unextendible product bases and bound entanglement},\ }\href {https://doi.org/10.1103/PhysRevLett.82.5385} {\bibfield  {journal} {\bibinfo  {journal} {Phys. Rev. Lett.}\ }\textbf {\bibinfo {volume} {82}},\ \bibinfo {pages} {5385} (\bibinfo {year} {1999})}\BibitemShut {NoStop}%
\bibitem [{\citenamefont {DiVincenzo}\ \emph {et~al.}(2003)\citenamefont {DiVincenzo}, \citenamefont {Mor}, \citenamefont {Shor}, \citenamefont {Smolin},\ and\ \citenamefont {Terhal}}]{DiVincenzoEtAl2003}%
  \BibitemOpen
  \bibfield  {author} {\bibinfo {author} {\bibfnamefont {D.~P.}\ \bibnamefont {DiVincenzo}}, \bibinfo {author} {\bibfnamefont {T.}~\bibnamefont {Mor}}, \bibinfo {author} {\bibfnamefont {P.~W.}\ \bibnamefont {Shor}}, \bibinfo {author} {\bibfnamefont {J.~A.}\ \bibnamefont {Smolin}},\ and\ \bibinfo {author} {\bibfnamefont {B.~M.}\ \bibnamefont {Terhal}},\ }\bibfield  {title} {\bibinfo {title} {Unextendible product bases, uncompletable product bases and bound entanglement},\ }\href {https://doi.org/10.1007/s00220-003-0877-6} {\bibfield  {journal} {\bibinfo  {journal} {Communications in Mathematical Physics}\ }\textbf {\bibinfo {volume} {238}},\ \bibinfo {pages} {379} (\bibinfo {year} {2003})}\BibitemShut {NoStop}%
\bibitem [{\citenamefont {Campbell}\ and\ \citenamefont {Browne}(2010)}]{CampbellBrowne2010}%
  \BibitemOpen
  \bibfield  {author} {\bibinfo {author} {\bibfnamefont {E.~T.}\ \bibnamefont {Campbell}}\ and\ \bibinfo {author} {\bibfnamefont {D.~E.}\ \bibnamefont {Browne}},\ }\bibfield  {title} {\bibinfo {title} {Bound states for magic state distillation in fault-tolerant quantum computation},\ }\href {https://doi.org/10.1103/PhysRevLett.104.030503} {\bibfield  {journal} {\bibinfo  {journal} {Physical Review Letters}\ }\textbf {\bibinfo {volume} {104}},\ \bibinfo {pages} {030503} (\bibinfo {year} {2010})}\BibitemShut {NoStop}%
\bibitem [{\citenamefont {Kwon}(2025)}]{Kwon2025}%
  \BibitemOpen
  \bibfield  {author} {\bibinfo {author} {\bibfnamefont {H.}~\bibnamefont {Kwon}},\ }\href@noop {} {\bibinfo {title} {Nonstabilizerness without magic: Classically simulatable quantum states that are indistinguishable by classically simulatable quantum circuits}} (\bibinfo {year} {2025}),\ \Eprint {https://arxiv.org/abs/2509.25790} {arXiv:2509.25790 [quant-ph]} \BibitemShut {NoStop}%
\bibitem [{\citenamefont {Gottesman}(1999)}]{Gottesman1999}%
  \BibitemOpen
  \bibfield  {author} {\bibinfo {author} {\bibfnamefont {D.}~\bibnamefont {Gottesman}},\ }\bibfield  {title} {\bibinfo {title} {Fault-tolerant quantum computation with higher-dimensional systems},\ }in\ \href@noop {} {\emph {\bibinfo {booktitle} {Quantum Computing and Quantum Communications}}},\ \bibinfo {editor} {edited by\ \bibinfo {editor} {\bibfnamefont {C.~P.}\ \bibnamefont {Williams}}}\ (\bibinfo  {publisher} {Springer Berlin Heidelberg},\ \bibinfo {address} {Berlin, Heidelberg},\ \bibinfo {year} {1999})\ pp.\ \bibinfo {pages} {302--313}\BibitemShut {NoStop}%
\bibitem [{\citenamefont {Howard}\ \emph {et~al.}(2013)\citenamefont {Howard}, \citenamefont {Brennan},\ and\ \citenamefont {Vala}}]{HowardBrennanVala2015}%
  \BibitemOpen
  \bibfield  {author} {\bibinfo {author} {\bibfnamefont {M.}~\bibnamefont {Howard}}, \bibinfo {author} {\bibfnamefont {E.}~\bibnamefont {Brennan}},\ and\ \bibinfo {author} {\bibfnamefont {J.}~\bibnamefont {Vala}},\ }\bibfield  {title} {\bibinfo {title} {Quantum contextuality with stabilizer states},\ }\href {https://doi.org/10.3390/e15062340} {\bibfield  {journal} {\bibinfo  {journal} {Entropy}\ }\textbf {\bibinfo {volume} {15}},\ \bibinfo {pages} {2340} (\bibinfo {year} {2013})}\BibitemShut {NoStop}%
\bibitem [{\citenamefont {Raussendorf}(2019)}]{Raussendorf2019}%
  \BibitemOpen
  \bibfield  {author} {\bibinfo {author} {\bibfnamefont {R.}~\bibnamefont {Raussendorf}},\ }\bibfield  {title} {\bibinfo {title} {Cohomological framework for contextual quantum computations},\ }\href {https://doi.org/10.48550/arXiv.1602.04155} {\bibfield  {journal} {\bibinfo  {journal} {Quantum Information and Computation}\ }\textbf {\bibinfo {volume} {19}},\ \bibinfo {pages} {1141} (\bibinfo {year} {2019})}\BibitemShut {NoStop}%
\bibitem [{\citenamefont {Raussendorf}\ \emph {et~al.}(2023)\citenamefont {Raussendorf}, \citenamefont {Okay}, \citenamefont {Zurel},\ and\ \citenamefont {Feldmann}}]{RaussendorfEtAl2023}%
  \BibitemOpen
  \bibfield  {author} {\bibinfo {author} {\bibfnamefont {R.}~\bibnamefont {Raussendorf}}, \bibinfo {author} {\bibfnamefont {C.}~\bibnamefont {Okay}}, \bibinfo {author} {\bibfnamefont {M.}~\bibnamefont {Zurel}},\ and\ \bibinfo {author} {\bibfnamefont {P.}~\bibnamefont {Feldmann}},\ }\bibfield  {title} {\bibinfo {title} {The role of cohomology in quantum computation with magic states},\ }\href {https://doi.org/10.22331/q-2023-04-13-979} {\bibfield  {journal} {\bibinfo  {journal} {{Quantum}}\ }\textbf {\bibinfo {volume} {7}},\ \bibinfo {pages} {979} (\bibinfo {year} {2023})}\BibitemShut {NoStop}%
\bibitem [{\citenamefont {{Gross}}(2006)}]{Gross2006}%
  \BibitemOpen
  \bibfield  {author} {\bibinfo {author} {\bibfnamefont {D.}~\bibnamefont {{Gross}}},\ }\bibfield  {title} {\bibinfo {title} {{Hudson's theorem for finite-dimensional quantum systems}},\ }\href {https://doi.org/10.1063/1.2393152} {\bibfield  {journal} {\bibinfo  {journal} {J. Math. Phys.}\ }\textbf {\bibinfo {volume} {47}},\ \bibinfo {pages} {122107} (\bibinfo {year} {2006})}\BibitemShut {NoStop}%
\bibitem [{\citenamefont {Peres}(1996)}]{Peres1996}%
  \BibitemOpen
  \bibfield  {author} {\bibinfo {author} {\bibfnamefont {A.}~\bibnamefont {Peres}},\ }\bibfield  {title} {\bibinfo {title} {Separability criterion for density matrices},\ }\href {https://doi.org/10.1103/PhysRevLett.77.1413} {\bibfield  {journal} {\bibinfo  {journal} {Phys. Rev. Lett.}\ }\textbf {\bibinfo {volume} {77}},\ \bibinfo {pages} {1413} (\bibinfo {year} {1996})}\BibitemShut {NoStop}%
\bibitem [{\citenamefont {Horodecki}\ \emph {et~al.}(1998)\citenamefont {Horodecki}, \citenamefont {Horodecki},\ and\ \citenamefont {Horodecki}}]{Horodeckisz1998}%
  \BibitemOpen
  \bibfield  {author} {\bibinfo {author} {\bibfnamefont {M.}~\bibnamefont {Horodecki}}, \bibinfo {author} {\bibfnamefont {P.}~\bibnamefont {Horodecki}},\ and\ \bibinfo {author} {\bibfnamefont {R.}~\bibnamefont {Horodecki}},\ }\bibfield  {title} {\bibinfo {title} {Mixed-state entanglement and distillation: Is there a ``bound'' entanglement in nature?},\ }\href {https://doi.org/10.1103/PhysRevLett.80.5239} {\bibfield  {journal} {\bibinfo  {journal} {Phys. Rev. Lett.}\ }\textbf {\bibinfo {volume} {80}},\ \bibinfo {pages} {5239} (\bibinfo {year} {1998})}\BibitemShut {NoStop}%
\bibitem [{\citenamefont {Johnston}(2014)}]{Johnston2014}%
  \BibitemOpen
  \bibfield  {author} {\bibinfo {author} {\bibfnamefont {N.}~\bibnamefont {Johnston}},\ }\bibfield  {title} {\bibinfo {title} {The structure of qubit unextendible product bases},\ }\href {https://doi.org/10.1088/1751-8113/47/42/424034} {\bibfield  {journal} {\bibinfo  {journal} {Journal of Physics A: Mathematical and Theoretical}\ }\textbf {\bibinfo {volume} {47}},\ \bibinfo {pages} {424034} (\bibinfo {year} {2014})},\ \Eprint {https://arxiv.org/abs/1401.7920} {arXiv:1401.7920 [quant-ph]} \BibitemShut {NoStop}%
\bibitem [{Note1()}]{Note1}%
  \BibitemOpen
  \bibinfo {note} {Write $S_4=\{\ket {u_i}\}_{i=1}^{8}$ in the order listed in Thm.~\ref {thm: 4 qubit USB}, and suppose $\ket {a_1}\protect \!\ket {a_2}\protect \!\ket {a_3}\protect \!\ket {a_4}$ were orthogonal to every member of $S_4$. Assign to each $i$ a qubit $\varphi (i)$ on which the local overlap vanishes. Since a qubit state has a unique orthogonal state, all members of a block $\varphi ^{-1}(j)$ must carry the \protect \emph {same} local factor on qubit $j$. Grouping the elements of $S_4$ by equal local factor on qubit $j$ then gives: \begin {align*} j=1,2:&\hskip 1em\relax \{1,2\},\{3,4\},\{5,6\},\{7,8\}\\ j=3:&\hskip 1em\relax \{1,7\},\{2,8\},\{3\},\{4\},\{5\},\{6\}\\ j=4:&\hskip 1em\relax \{1\},\{2\},\{3,5\},\{4,6\},\{7\},\{8\} \end {align*} Note that every block has at most two elements, and since $|S_4|=8$ and there are four blocks, all four blocks must have exactly two elements. Since $\varphi ^{-1}(3)\in \{\{1,7\},\{2,8\}\}$ and $\varphi ^{-1}(4)\in \{\{3,5\},\{4,6\}\}$, the four resulting
  choices leave the residues $\{1,3,5,7\}$, $\{1,4,6,7\}$, $\{2,3,5,8\}$ and $\{2,4,6,8\}$, yet none of which contains an admissible block for qubit $1$ or $2$. Hence, no product vector is orthogonal to all elements of $S_4$, that is, $S_4$ is a UPB.}\BibitemShut {Stop}%
\bibitem [{\citenamefont {Taylor}(1992)}]{Taylor1992}%
  \BibitemOpen
  \bibfield  {author} {\bibinfo {author} {\bibfnamefont {D.~E.}\ \bibnamefont {Taylor}},\ }\href@noop {} {\emph {\bibinfo {title} {The Geometry of the Classical Groups}}},\ \bibinfo {series} {Sigma Series in Pure Mathematics}, Vol.~\bibinfo {volume} {9}\ (\bibinfo  {publisher} {Heldermann Verlag},\ \bibinfo {address} {Berlin},\ \bibinfo {year} {1992})\BibitemShut {NoStop}%
\bibitem [{\citenamefont {Hirschfeld}(1998)}]{Hirschfeld1998}%
  \BibitemOpen
  \bibfield  {author} {\bibinfo {author} {\bibfnamefont {J.}~\bibnamefont {Hirschfeld}},\ }\href@noop {} {\emph {\bibinfo {title} {Projective Geometries Over Finite Fields}}},\ Oxford mathematical monographs\ (\bibinfo  {publisher} {Clarendon Press},\ \bibinfo {year} {1998})\BibitemShut {NoStop}%
\bibitem [{Note2()}]{Note2}%
  \BibitemOpen
  \bibinfo {note} {This ovoid corresponds with maximal sets (with cardinality $5$) of mutually anti-commuting Pauli operators.}\BibitemShut {Stop}%
\bibitem [{\citenamefont {L{\'e}vay}\ \emph {et~al.}(2013)\citenamefont {L{\'e}vay}, \citenamefont {Planat},\ and\ \citenamefont {Saniga}}]{LevayPlanatSangia2013}%
  \BibitemOpen
  \bibfield  {author} {\bibinfo {author} {\bibfnamefont {P.}~\bibnamefont {L{\'e}vay}}, \bibinfo {author} {\bibfnamefont {M.}~\bibnamefont {Planat}},\ and\ \bibinfo {author} {\bibfnamefont {M.}~\bibnamefont {Saniga}},\ }\bibfield  {title} {\bibinfo {title} {Grassmannian connection between three- and four-qubit observables, mermin's contextuality and black holes},\ }\href {https://doi.org/10.1007/JHEP09(2013)037} {\bibfield  {journal} {\bibinfo  {journal} {Journal of High Energy Physics}\ }\textbf {\bibinfo {volume} {2013}},\ \bibinfo {pages} {37} (\bibinfo {year} {2013})}\BibitemShut {NoStop}%
\bibitem [{\citenamefont {van Geemen}\ and\ \citenamefont {Marrani}(2019)}]{VanGeemenMarrani2019}%
  \BibitemOpen
  \bibfield  {author} {\bibinfo {author} {\bibfnamefont {B.}~\bibnamefont {van Geemen}}\ and\ \bibinfo {author} {\bibfnamefont {A.}~\bibnamefont {Marrani}},\ }\bibfield  {title} {\bibinfo {title} {Lagrangian grassmannians and spinor varieties in characteristic two},\ }\bibfield  {journal} {\bibinfo  {journal} {Symmetry, Integrability and Geometry: Methods and Applications}\ }\href {https://doi.org/10.3842/sigma.2019.064} {10.3842/sigma.2019.064} (\bibinfo {year} {2019})\BibitemShut {NoStop}%
\bibitem [{\citenamefont {Miezaki}\ and\ \citenamefont {Munemasa}(2024)}]{MiezakiMunemasa2024}%
  \BibitemOpen
  \bibfield  {author} {\bibinfo {author} {\bibfnamefont {T.}~\bibnamefont {Miezaki}}\ and\ \bibinfo {author} {\bibfnamefont {A.}~\bibnamefont {Munemasa}},\ }\bibfield  {title} {\bibinfo {title} {Jacobi polynomials and harmonic weight enumerators of the first-order reed--muller codes and the extended hamming codes},\ }\href {https://doi.org/10.1007/s10623-023-01327-0} {\bibfield  {journal} {\bibinfo  {journal} {Designs, Codes and Cryptography}\ }\textbf {\bibinfo {volume} {92}},\ \bibinfo {pages} {1041} (\bibinfo {year} {2024})}\BibitemShut {NoStop}%
\bibitem [{\citenamefont {Conway}\ \emph {et~al.}(2013)\citenamefont {Conway}, \citenamefont {Sloane}, \citenamefont {Bannai}, \citenamefont {Borcherds}, \citenamefont {Leech}, \citenamefont {Norton}, \citenamefont {Odlyzko}, \citenamefont {Parker}, \citenamefont {Queen},\ and\ \citenamefont {Venkov}}]{ConwaySloane2013}%
  \BibitemOpen
  \bibfield  {author} {\bibinfo {author} {\bibfnamefont {J.}~\bibnamefont {Conway}}, \bibinfo {author} {\bibfnamefont {N.}~\bibnamefont {Sloane}}, \bibinfo {author} {\bibfnamefont {E.}~\bibnamefont {Bannai}}, \bibinfo {author} {\bibfnamefont {R.}~\bibnamefont {Borcherds}}, \bibinfo {author} {\bibfnamefont {J.}~\bibnamefont {Leech}}, \bibinfo {author} {\bibfnamefont {S.}~\bibnamefont {Norton}}, \bibinfo {author} {\bibfnamefont {A.}~\bibnamefont {Odlyzko}}, \bibinfo {author} {\bibfnamefont {R.}~\bibnamefont {Parker}}, \bibinfo {author} {\bibfnamefont {L.}~\bibnamefont {Queen}},\ and\ \bibinfo {author} {\bibfnamefont {B.}~\bibnamefont {Venkov}},\ }\href@noop {} {\emph {\bibinfo {title} {Sphere Packings, Lattices and Groups}}},\ Grundlehren der mathematischen Wissenschaften\ (\bibinfo  {publisher} {Springer New York},\ \bibinfo {year} {2013})\BibitemShut {NoStop}%
\bibitem [{\citenamefont {Elkies}(2004)}]{Elkies2004}%
  \BibitemOpen
  \bibfield  {author} {\bibinfo {author} {\bibfnamefont {N.~D.}\ \bibnamefont {Elkies}},\ }\bibfield  {title} {\bibinfo {title} {Yet another proof of the uniqueness of the {$E_8$} lattice}} (\bibinfo {year} {2004}),\ \bibinfo {note} {unpublished note}\BibitemShut {NoStop}%
\bibitem [{\citenamefont {Ruuge}\ and\ \citenamefont {van Oystaeyen}(2005)}]{RuugeVanOystaeyen2005}%
  \BibitemOpen
  \bibfield  {author} {\bibinfo {author} {\bibfnamefont {A.~E.}\ \bibnamefont {Ruuge}}\ and\ \bibinfo {author} {\bibfnamefont {F.}~\bibnamefont {van Oystaeyen}},\ }\bibfield  {title} {\bibinfo {title} {Saturated kochen–specker-type configuration of 120 projective lines in eight-dimensional space and its group of symmetry},\ }\href {https://doi.org/10.1063/1.1887923} {\bibfield  {journal} {\bibinfo  {journal} {Journal of Mathematical Physics}\ }\textbf {\bibinfo {volume} {46}},\ \bibinfo {pages} {052109} (\bibinfo {year} {2005})}\BibitemShut {NoStop}%
\bibitem [{\citenamefont {Waegell}\ and\ \citenamefont {Aravind}(2015)}]{WaegellAravind2015}%
  \BibitemOpen
  \bibfield  {author} {\bibinfo {author} {\bibfnamefont {M.}~\bibnamefont {Waegell}}\ and\ \bibinfo {author} {\bibfnamefont {P.~K.}\ \bibnamefont {Aravind}},\ }\bibfield  {title} {\bibinfo {title} {Parity proofs of the kochen–specker theorem based on the lie algebra e8},\ }\href {https://doi.org/10.1088/1751-8113/48/22/225301} {\bibfield  {journal} {\bibinfo  {journal} {Journal of Physics A: Mathematical and Theoretical}\ }\textbf {\bibinfo {volume} {48}},\ \bibinfo {pages} {225301} (\bibinfo {year} {2015})}\BibitemShut {NoStop}%
\bibitem [{\citenamefont {Mari}\ and\ \citenamefont {Eisert}(2012)}]{MariEisert2012}%
  \BibitemOpen
  \bibfield  {author} {\bibinfo {author} {\bibfnamefont {A.}~\bibnamefont {Mari}}\ and\ \bibinfo {author} {\bibfnamefont {J.}~\bibnamefont {Eisert}},\ }\bibfield  {title} {\bibinfo {title} {Positive wigner functions render classical simulation of quantum computation efficient},\ }\href {https://doi.org/10.1103/PhysRevLett.109.230503} {\bibfield  {journal} {\bibinfo  {journal} {Phys. Rev. Lett.}\ }\textbf {\bibinfo {volume} {109}},\ \bibinfo {pages} {230503} (\bibinfo {year} {2012})}\BibitemShut {NoStop}%
\bibitem [{\citenamefont {Howard}\ \emph {et~al.}(2014)\citenamefont {Howard}, \citenamefont {Wallman}, \citenamefont {Veitch},\ and\ \citenamefont {Emerson}}]{HowardEtAl2014}%
  \BibitemOpen
  \bibfield  {author} {\bibinfo {author} {\bibfnamefont {M.}~\bibnamefont {Howard}}, \bibinfo {author} {\bibfnamefont {J.}~\bibnamefont {Wallman}}, \bibinfo {author} {\bibfnamefont {V.}~\bibnamefont {Veitch}},\ and\ \bibinfo {author} {\bibfnamefont {J.}~\bibnamefont {Emerson}},\ }\bibfield  {title} {\bibinfo {title} {Contextuality supplies the `magic' for quantum computation},\ }\href {https://doi.org/10.1038/nature13460} {\bibfield  {journal} {\bibinfo  {journal} {Nature}\ }\textbf {\bibinfo {volume} {510}},\ \bibinfo {pages} {351—355} (\bibinfo {year} {2014})}\BibitemShut {NoStop}%
\bibitem [{\citenamefont {Delfosse}\ \emph {et~al.}(2015)\citenamefont {Delfosse}, \citenamefont {Allard~Guerin}, \citenamefont {Bian},\ and\ \citenamefont {Raussendorf}}]{DelfosseEtAl2015}%
  \BibitemOpen
  \bibfield  {author} {\bibinfo {author} {\bibfnamefont {N.}~\bibnamefont {Delfosse}}, \bibinfo {author} {\bibfnamefont {P.}~\bibnamefont {Allard~Guerin}}, \bibinfo {author} {\bibfnamefont {J.}~\bibnamefont {Bian}},\ and\ \bibinfo {author} {\bibfnamefont {R.}~\bibnamefont {Raussendorf}},\ }\bibfield  {title} {\bibinfo {title} {Wigner function negativity and contextuality in quantum computation on rebits},\ }\href {https://doi.org/10.1103/PhysRevX.5.021003} {\bibfield  {journal} {\bibinfo  {journal} {Phys. Rev. X}\ }\textbf {\bibinfo {volume} {5}},\ \bibinfo {pages} {021003} (\bibinfo {year} {2015})}\BibitemShut {NoStop}%
\bibitem [{\citenamefont {Zurel}\ \emph {et~al.}(2020)\citenamefont {Zurel}, \citenamefont {Okay},\ and\ \citenamefont {Raussendorf}}]{ZurelOkayRaussendorf2020}%
  \BibitemOpen
  \bibfield  {author} {\bibinfo {author} {\bibfnamefont {M.}~\bibnamefont {Zurel}}, \bibinfo {author} {\bibfnamefont {C.}~\bibnamefont {Okay}},\ and\ \bibinfo {author} {\bibfnamefont {R.}~\bibnamefont {Raussendorf}},\ }\bibfield  {title} {\bibinfo {title} {Hidden variable model for universal quantum computation with magic states on qubits},\ }\href {https://doi.org/10.1103/PhysRevLett.125.260404} {\bibfield  {journal} {\bibinfo  {journal} {Phys. Rev. Lett.}\ }\textbf {\bibinfo {volume} {125}},\ \bibinfo {pages} {260404} (\bibinfo {year} {2020})}\BibitemShut {NoStop}%
\end{thebibliography}%

\appendix
\onecolumngrid

\appendix

\section{Proof of Thm.~\ref{thm: 4 qubit USB}}\label{app: proof - 4 qubits}

For our main result, we will use the following lemma.

\begin{lemma}\label{lm: 3-qubit obstruction}
    No $3$-qubit stabiliser state is orthogonal to every member of the following set of $3$-qubit stabiliser states:
    \begin{align*}
        \cS_3:=
        \bigl\{
        \ket{z_+z_+z_+},\ 
         \ket{z_+z_+z_-},\ 
         \ket{z_-x_+x_+},\ 
        \ket{x_+z_-x_-},\ 
         \ket{x_-x_-z_+},\ 
         \ket{x_-x_-z_-}
        \bigl\}.
    \end{align*}
\end{lemma}

\begin{proof}
    Assume that $\ket{\psi}$ is a stabiliser state orthogonal to every element of $\cS_3$, and denote by $A_\psi$ its stabiliser group.
    
    Recall that the codespace projector of a stabiliser group $A$ of a stabiliser code is given by $\Pi_A=\frac{1}{|A|}\sum_{P\in A}P$, and that for a pure stabiliser state (with $A<\cP_n$ maximal), the expectation of every Pauli operator takes values in $\{0,\pm1\}$ (it is either an eigenstate or an equal-weighted superposition). Hence, if $\bra{\psi}\Pi_A\ket{\psi}=0$ then at least one non-identity Pauli operator $P\in A$ has expectation $-1$, so that $-P\in A_\psi$.
    
    Now, the first two elements of $\cS_3$ span the code $\ket{z_+z_+}_{12}\otimes\mathbb C^2_3$, whose stabiliser group is generated by $Z_1$ and $Z_2$. It follows that $A_\psi$ contains one of the following elements
    \begin{align}\label{eq: Z stabilisers}
        -Z_1,\qquad -Z_2,\qquad -Z_1Z_2\; .
    \end{align}
    Similarly, the last two elements of $\cS_3$ span $\ket{x_-x_-}_{12}\otimes\mathbb C^2_3$, whose stabiliser group is generated by $-X_1$ and $-X_2$. Consequently, $A_\psi$ also contains an element from
    \begin{align}\label{eq: X stabilisers}
        X_1,\qquad X_2,\qquad -X_1X_2\; .
    \end{align}
    Since $A_\psi$ is Abelian, the respective elements from Eq.~(\ref{eq: Z stabilisers}) and Eq.~(\ref{eq: X stabilisers}) must commute; this leaves three possibilities, 
    \begin{align*}
        (-Z_1,+X_2),\qquad
        (-Z_2,+X_1),\qquad
        (-Z_1Z_2,-X_1X_2)\; .
    \end{align*}
    In the first case, $\ket{\psi}=\ket{z_-x_+}_{12}\otimes\ket r_3$ for $r_3\in\{x_\pm,y_\pm,z_\pm\}$. Orthogonality to $\ket{z_-x_+x_+}$ forces $\ket r=\ket{x_-}$, and thus
    \begin{align*}
        \bracket{x_+z_-x_-}{\psi}
        =
        \bracket{x_+}{z_-}\,
        \bracket{z_-}{x_+}\,
        \bracket{x_-}{x_-}
        \neq0,
    \end{align*}
    contradicting our assumption on $\ket{\psi}$. The second case is analogous: one has $\ket{\psi}=\ket{x_+z_-}_{12}\otimes\ket r_3$ for $r_3\in\{x_\pm,y_\pm,z_\pm\}$, and orthogonality to $\ket{x_+z_-x_-}$ forces $\ket r=\ket{x_+}$, which implies $\bracket{z_-x_+x_+}{\psi}\neq0$.

    In the final case, the first two qubits are stabilised by $-Z_1Z_2$ and $-X_1X_2$, hence,
    \begin{align*}
        \ket{\psi}
        &=\ket{\psi^-}_{12}\otimes\ket r_3\; , &
        \ket{\psi^-}
        &:=\frac{\ket{01}-\ket{10}}{\sqrt2}\; .
    \end{align*}
    Note that $\bracket{z_-x_+}{\psi^-}\neq0$ and $\bracket{x_+z_-}{\psi^-}\neq0$. Orthogonality to $\ket{z_-x_+x_+}$ therefore requires $\ket r=\ket{x_-}$, whereas orthogonality to $\ket{x_+z_-x_-}$ requires $\ket r=\ket{x_+}$, which together is impossible.
\end{proof}

Using the above, we now prove our main assertion.

\begin{proof}[Proof of Thm.~\ref{thm: 4 qubit USB}]
    Each element of $\cS_4$ is a product of single-qubit Pauli eigenstates and is therefore a stabiliser state. The states are pairwise orthogonal, as every pair has a tensor factor with opposite eigenstates of some Pauli operator.
    
    Suppose, for contradiction, that a stabiliser state $\ket{\psi}$ is orthogonal to every element of $\cS_4$. Let $V=\zz^8_2$ be the binary symplectic space of four qubits, with symplectic form $\omega$. Write $x_j,z_j\in V$ for the classes of $X_j,Z_j$, and put $y_j=x_j+z_j$. Let further $L\subset V$ be the four-dimensional Lagrangian subspace obtained from the unsigned stabiliser group of $\ket{\psi}$.
    
    Note that the Lagrangian subspaces corresponding to the eight states in $\cS_4$ are
    \begin{align*}
        L_1
        &=\langle z_1,z_2,z_3,z_4\rangle\; ,
        &
        L_2
        &=\langle z_1,z_2,z_3,x_4\rangle\; ,&
        L_3
        &=\langle z_1,x_2,x_3,y_4\rangle\; ,&
        L_4
        &=\langle z_1,x_2,y_3,y_4\rangle\; ,\\
        L_5
        &=\langle x_1,z_2,x_3,y_4\rangle\; ,&
        L_6
        &=\langle x_1,z_2,y_3,y_4\rangle\; ,&
        L_7
        &=\langle x_1,x_2,z_3,z_4\rangle\; ,&
        L_8
        &=\langle x_1,x_2,z_3,x_4\rangle\; ,
    \end{align*}
    where $x_i,y_i,z_i$ denote the symplectic vectors $v\in V$ of the corresponding Weyl operators $W_v$. 
    
    Note that the formula for the overlap of stabiliser states in Eq.~(\ref{eq: stabiliser overlap}) in particular implies that
    \begin{align*}
        \bracket{\psi}{\psi_i}=0
        \quad\Longrightarrow\quad
        L\cap L_i\neq\{0\}\; .
    \end{align*}
    
    We will also need the following observation. Let $U\subset V$ be isotropic and suppose
    \begin{align*}
        C
        &=U\oplus\langle v\rangle\; ,&
        C'
        &=U\oplus\langle w\rangle\; ,
    \end{align*}
    where $0\neq v,w\in U^\perp$ and $\omega(v,w)=1$. If a Lagrangian $L\in\Lag(V)$ intersects both $C$ and $C'$ nontrivially, then $L\cap U\neq\{0\}$. Indeed, suppose that $L\cap U=\{0\}$, then any nonzero vectors in $L\cap C$ and $L\cap C'$ are of the form $u+v$ and $u'+w$ for $u,u'\in U$. Since $U$ is isotropic, and $v,w\in U^\perp$, $\omega(u+v,u'+w)=\omega(v,w)=1$, contradicting that $L$ is isotropic.
    
    Now, apply this to the pairs
    $(L_1,L_2)$, $(L_7,L_8)$, $(L_3,L_4)$ and $(L_5,L_6)$ to conclude
    \begin{align}\label{eq: four intersections}
        L&\cap\langle z_1,z_2,z_3\rangle\neq\{0\}\; , &
        L&\cap\langle x_1,x_2,z_3\rangle\neq\{0\}\; , &
        L&\cap\langle z_1,x_2,y_4\rangle\neq\{0\}\; , &
        L&\cap\langle x_1,z_2,y_4\rangle\neq\{0\}\; .
    \end{align}
    
    We claim that this implies that either $z_3\in L$ or $y_4\in L$. Suppose otherwise, and choose nonzero vectors
    \begin{align*}
        a&=\alpha z_1+\beta z_2+\gamma z_3\; ,&
        b&=\delta x_1+\varepsilon x_2+\zeta z_3\; ,&
        d&=\eta z_1+\theta x_2+\iota y_4\; ,&
        e&=\kappa x_1+\lambda z_2+\mu y_4\; .
    \end{align*}
    in the respective four intersections
    in Eq.~(\ref{eq: four intersections}) (with Boolean coefficients). Since $z_3,y_4\notin L$, each of the pairs $(\alpha,\beta)$, $(\delta,\varepsilon)$, $(\eta,\theta)$ and $(\kappa,\lambda)$ must be nonzero. Isotropy of $L$ thus implies
    \begin{align}\label{eq: coefficients constrained}
        \omega(a,b)&=\alpha\delta+\beta\varepsilon=0\;, &
        \omega(a,d)&=\beta\theta=0\; ,&
        \omega(a,e)&=\alpha\kappa=0\; ,&
        \omega(b,d)&=\delta\eta=0\; ,&
        \omega(b,e)&=\varepsilon\lambda=0\; .
    \end{align}
    We evaluate these constraints for the three possible nonzero values of $(\alpha,\beta)$:
    
    if $(\alpha,\beta)=(1,0)$, then Eq.~(\ref{eq: coefficients constrained}) gives $\delta=0$ and thus $\varepsilon=1$ from which $\lambda=0$ and thus $\kappa=1$, contradicting $\alpha\kappa=0$;
    
    if $(\alpha,\beta)=(0,1)$, then Eq.~(\ref{eq: coefficients constrained}) gives $\varepsilon=0$ and thus $\delta=1$ from which $\eta=0$ and thus $\theta=1$, contradicting $\beta\theta=0$;
    
    if $(\alpha,\beta)=(1,1)$, then Eq.~(\ref{eq: coefficients constrained}) gives $\theta=0$ and thus $\eta=1$ from which $\delta=0$ and thus $\varepsilon=1$, contradicting $\alpha\delta+\beta\varepsilon=0$.
    
    \noindent Consequently, we must have either $z_3\in L$ or $y_4\in L$.

    If $z_3\in L$, then the stabiliser group of $\ket{\psi}$
    contains either $+Z_3$ or $-Z_3$, and hence
    \begin{align*}
        \ket{\psi}
        =\ket{\psi}_{124}\otimes\ket{z_s}_3\; ,
        \qquad s\in\{+,-\}\; ,
    \end{align*}
    up to the ordering of tensor factors. Likewise, if $y_4\in L$, then
    \begin{align*}
        \ket{\psi}
        =\ket{\psi}_{123}\otimes\ket{y_s}_4,
        \qquad s\in\{+,-\}\; .
    \end{align*}
    Note that for any of these two fixed one-qubit eigenstate, there are exactly two elements of $\cS_4$ orthogonal to $\ket{\psi}$ on that qubit. Since the inner product of product states factorises, orthogonality to the remaining six states is equivalent to orthogonality of $\ket{\psi}$ to the corresponding residual $3$-qubit product states.
    
    For $Z_3=\pm1$, in the qubit order $(1,2,4)$, the residual sets are
    \begin{align*}
        \mathcal R_{Z_3,+}
        &:=\bigl\{
         \ket{z_+z_+z_+},\  
         \ket{z_-x_+y_+},\ 
         \ket{z_-x_+y_-},\ 
         \ket{x_+z_-y_+},\ 
         \ket{x_+z_-y_-},\ 
         \ket{x_-x_-z_-}
        \bigr\}\; ,\\
        \mathcal R_{Z_3,-}
        &:=\bigl\{
         \ket{z_+z_+x_+},\ 
         \ket{z_-x_+y_+},\ 
         \ket{z_-x_+y_-},\ 
         \ket{x_+z_-y_+},\ 
         \ket{x_+z_-y_-},\ 
         \ket{x_-x_-x_-}
        \bigr\}\; .
    \end{align*}
    Similarly, for $Y_4=\pm 1$, in the qubit order $(1,2,3)$, the residual sets are
    \begin{align*}
        \mathcal R_{Y_4,+}
        &:=
        \bigl\{
        \ket{z_+z_+z_+},\ 
         \ket{z_+z_+z_-},\ 
         \ket{z_-x_+x_+},\ 
        \ket{x_+z_-x_-},\ 
         \ket{x_-x_-z_+},\ 
         \ket{x_-x_-z_-}
        \bigr\}\; \\
        \mathcal R_{Y_4,-}
        &:=\bigl\{
         \ket{z_+z_+z_+},\ 
         \ket{z_+z_+z_-},\ 
         \ket{z_-x_+y_-},\ 
         \ket{x_+z_-y_+},\ 
         \ket{x_-x_-z_+},\ 
         \ket{x_-x_-z_-}
        \bigr\}\; .
    \end{align*}
    Note that $R_{Y_4,+}=\cS_3$ from Lm.~\ref{lm: 3-qubit obstruction}. Moreover, the other sets are locally Clifford equivalent to $\cS_3$, explicitly
    \begin{align*}
        (Y\otimes H\otimes C)\mathcal R_{Z_3,+}
        &=\cS_3\; ,&
        (Y\otimes H\otimes R_X)\mathcal R_{Z_3,-}
        &=\cS_3\; ,&
        (I\otimes I\otimes S)\mathcal R_{Y_4,-}
        &=\cS_3\; ,
    \end{align*}
    where $S=\mathrm{diag}(1,i)$ and $H=\frac{1}{\sqrt{2}}\begin{pmatrix}
        1 & 1 \\ 1 & -1
    \end{pmatrix}$ are the phase and Hadamard gates, and $C=HS^\dagger$ and $R_X=\exp(-\frac{i\pi}{4}X)=\exp(-\frac{i\pi}{4})HSH$. Writing $s\in\{+,-\}$, their action on Pauli eigenstates is given, up to an overall phase, by
    \begin{align*}
    \begin{array}{c|ccc}
        U & \ket{x_s} & \ket{y_s} & \ket{z_s}\\[.1cm]
        \hline
        Y   & \ket{x_{-s}} & \ket{y_s}    & \ket{z_{-s}}\\[.1cm]
        H   & \ket{z_s}    & \ket{y_{-s}} & \ket{x_s}\\[.1cm]
        C   & \ket{y_s}    & \ket{z_s}    & \ket{x_s}\\[.1cm]
        R_X & \ket{x_s}    & \ket{z_s}    & \ket{y_{-s}}\\[.1cm]
        S   & \ket{y_s}    & \ket{x_{-s}} & \ket{z_s}
    \end{array}
    \end{align*}
    
    Clearly, local Clifford transformations preserve stabiliser states and their orthogonality relations. The existence of a state $\ket{\psi}$ orthogonal to every element in $\cS_4$ thus also implies the existence of a $3$-qubit stabiliser state orthogonal to every element of $\cS_3$, contradicting Lm.~\ref{lm: 3-qubit obstruction}.
    
    Consequently, no $4$-qubit stabiliser state is orthogonal to all elements of $\cS_4$, that is, $\cS_4$ is a maximal orthogonal set of stabiliser states of cardinality $|\cS_4|=8<16$. In particular, it cannot be extended to a four-qubit stabiliser basis.
\end{proof}

\section{Proof of Thm.~\ref{thm: 3 qudit USB}}\label{app: proof - 3 qudits}

Let $q$ be odd prime and let $V:=\mathbb F_q^{2n}$ be the $n$-qudit symplectic phase space, and recall that we identify pure stabiliser states with affine Lagrangian subspaces $A\subset V$. By Cor.~\ref{cor: affine stabiliser overlap}, their overlaps are of the form
\begin{align*}
    \left|\bracket{A}{B}\right|^2
    &=q^{-n}|A\cap B|\; ,&
    \bracket{A}{B}=0
    &\iff A\cap B=\emptyset\; .
\end{align*}
We will also use the fact that for any stabiliser subgroup $A$, the support of the outcome distribution of a stabiliser state (after measurement in the joint eigenbasis of $A$) is an affine subspace of the corresponding outcome space.

\begin{lemma}\label{lm: affine stabiliser support}
    Let $q$ be prime, $L,M\in\Lag(V)$ and
    $\lambda\in\Xi(L)$. In the basis
    $\{\ket{M,\eta}\mid\eta\in\Xi(M)\}$, the \emph{support of $\ket{L,\lambda}$} is
    \begin{align}\label{eq: outcome set}
        S_M(L,\lambda)
        =\{\eta\in\Xi(M)\mid\eta|_{L\cap M}=\lambda|_{L\cap M}\}\; .
    \end{align}
    It is an affine subspace of dimension
    $n-\dim(L\cap M)$, and the outcome distribution is uniform on
    this support.
\end{lemma}

\begin{proof}
    Eq.~(\ref{eq: outcome set}) follows directly from Lm.~\ref{lm: stabiliser overlap}. For odd $q$, if $\mathrm{res}_K:M^*\longrightarrow K^*$ denotes restriction to $K:=L\cap M$, then 
    \begin{align*}
        S_M(L,\lambda)
        =\mathrm{res}_K^{-1}(\lambda|_K)\; .
    \end{align*}
    Clearly, $S_M(L,\lambda)$ is an affine coset of $\ker(\mathrm{res}_K)=\left\{\eta\in M^*\mid\eta|_K=0\right\}$ of dimension $n-\dim(K)$. For $q=2$, fix any $\eta_0\in\Xi(M)$ such that $\eta=\eta_0+\xi$ for $\eta\in\Xi(M)$ with $\xi\in M^*$. Now, note that $\delta(\lambda|_K)=\delta(\eta_0|_K)=\beta_W|_{K\times K}$, from which $\lambda|_K-\eta_0|_K\in K^*$ and thus $S_M(L,\lambda)=\eta_0+\mathrm{res}_K^{-1}(\lambda|_K-\eta_0|_K)$ is an affine coset of dimension $n-\dim(K)$. Finally, uniformity follows from Eq.~(\ref{eq: stabiliser overlap}).
\end{proof}

In the affine-Lagrangian representation, this can also be seen as follows. Define $\pi_M:V\ra M^*$, $\pi_M(a)(v):=\omega(a,v)$. If $A=a+L$ represents the candidate stabiliser state, then
\begin{align*}
    S_M(A)=\pi_M(A)
    =\pi_M(a)+\pi_M(L)\; .
\end{align*}

\begin{lemma}\label{lm: two-qudit obstructions}
    Let $q$ be odd prime. Then no two-qudit stabiliser state is orthogonal to all elements of the following sets:
    \begin{equation}
    \begin{aligned}
        \cT_1
         &\ :=\ \bigl\{\ket{w_bw_c}\mid b,c\in\mathbb F_q^\times\bigr\}\ \cup\ \bigl\{\ket{w_0x_0},\ket{z_0w_0}\bigr\}\; ,\\
        \cT_2
         &\ :=\ \bigl\{\ket{x_0w_c}\mid c\in\mathbb F_q^\times\bigr\}\ \cup\ \bigl\{\ket{w_ax_0}\mid a\in\mathbb F_q^\times\bigr\}\ \cup\ \bigl\{\ket{w_0w_0}\bigr\}\; ,
        \\
        \cT_3
         &\ :=\ \bigl\{\ket{x_0w_b}\mid b\in\mathbb F_q^\times\bigr\}\ \cup\ \bigl\{\ket{w_aw_0}\mid a\in\mathbb F_q^\times\bigr\}\ \cup\ \bigl\{\ket{w_0z_0}\bigr\}\; .
    \end{aligned}
    \end{equation}
\end{lemma}

\begin{proof}
    Suppose that a stabiliser state is orthogonal to the elements in $\{\ket{w_bw_c}\mid b,c\neq0\}\subset\cT_1$. By Lm.~\ref{lm: affine stabiliser support}, its outcome support for the Lagrangian corresponding to the Abelian subgroup $\langle W_1,W_2\rangle$ is an affine subspace $S\subset\mathbb F_q^2$ satisfying
    \begin{align*}
        S\subseteq\{s_1=0\}\cup\{s_2=0\}\; .
    \end{align*}
    Since $2<q$, an affine subspace cannot be covered by these two proper affine hyperplanes unless it is contained in one of them. Thus the state is stabilised either by $W_1$ or by $W_2$ with eigenvalue
    label $0$, and hence factorises as either
    \begin{align*}
        \ket{w_0}\ket r
        \qquad\text{or}\qquad
        \ket r\ket{w_0}\; .
    \end{align*}
    In the first case, orthogonality to $\ket{w_0x_0}$ forces $\ket r\perp\ket{x_0}$, that is, $r=x_i$ for $i\neq 0$, after which its overlap with $\ket{z_0w_0}$ is nonzero. The second case is analogous. This proves the assertion for $\cT_1$.

    Before analysing the other two cases, observe that a two-qudit pure stabiliser $|L,\xi\rangle$ state is either a product or maximally entangled state. Indeed, let $L\subset V=V_A\oplus V_B$ be its Lagrangian subspace and set $L_A:=L\cap V_A$. Taking the partial trace over the second qudit gives $\rho_A=\frac{1}{q}\sum_{v\in L_A}\chi_\xi|_A(v)W_v^{(A)}$. Since $L_A$ is isotropic in the two-dimensional symplectic space $V_A$, its dimension is either zero or one. If $\dim L_A=0$, then $\rho_A=\frac{\one}{q}$ and $|L,\xi\rangle$ is maximally entangled, if $\dim L_A=1$, then $\rho_A$ is a rank-one one-qudit stabiliser projector, hence, $|L,\xi\rangle$ is a product state.
    
    Now, consider $\cT_2$. For a product stabiliser state $\ket r\ket s$, orthogonality to the first and second families in
    $\cT_2$ gives,
    \begin{align*}
        \ket r\perp\ket{x_0}\ &\text{or}\ \ket s=\ket{w_0}\; ,&
        \ket r=\ket{w_0}\ &\text{or}\ \ket s\perp\ket{x_0}\; ,
    \end{align*}
    respectively. The four possible combinations are either mutually inconsistent or give nonzero overlap with $\ket{w_0w_0}$.
    
    For a maximally entangled stabiliser state $\ket{\psi}$, let $U$ be a one-qudit Clifford unitary such that
    \begin{align*}
        \ket{\psi}=\ket{\psi_U}
         &=(I\otimes U)\ket{\Phi_q}\; ,&
        \ket{\Phi_q}
         &=\frac1{\sqrt q}\sum_{j\in\mathbb F_q}\ket{jj}\; .
    \end{align*}
    Since $\bracket{r,s}{\psi_U}=\frac{1}{\sqrt{q}}\bra sU\ket{\overline r}$, orthogonality to $\ket{w_ax_0}$ for every $a\neq0$ implies $U\ket{\overline{w_0}}\propto\ket{x_0}$. However, in this case, the overlap with the last vector in $\cT_2$ is nonzero, $\bracket{w_0w_0}{\psi_U}\propto\bracket{w_0}{x_0}\neq0$, hence, also $\cT_2$ admits no stabiliser extension.
    
    Analogously, for $\cT_3$, the product-state alternatives (for $\ket{r}\ket{s}$) are, respectively,
    \begin{align*}
        \ket r\perp\ket{x_0}\ &\text{or}\ \ket s=\ket{w_0}\; ,&
        \ket r=\ket{w_0}\ &\text{or}\ \ket s\perp\ket{w_0}\; .
    \end{align*}
    Yet, the singleton $\ket{w_0z_0}$ excludes all possible combinations. Finally, in the maximally entangled case, orthogonality to $\ket{w_aw_0}$ for all $a\neq0$ implies $U\ket{\overline{w_0}}\propto\ket{w_0}$, in which case $\bracket{w_0z_0}{\psi_U}\propto\bracket{z_0}{w_0}\neq0$.
\end{proof}

\begin{theorem}\label{thm: unextendible qudits, d>3}
    Let $q\geq5$ be prime. The following set is an unextendible stabiliser basis:
    \begin{equation}\label{eq: unextendible qudit stabiliser set}
    \begin{aligned}
        \cS_q:=&
         \bigl\{\ket{x_az_bx_c}\mid a,b,c\in\mathbb F_q^\times\bigr\}
        \ \cup\ 
         \bigl\{\ket{w_0z_0w_0}\bigr\}
        \ \cup\ 
         \bigl\{\ket{x_0w_bw_c}\mid b,c\in\mathbb F_q^\times\bigr\}
        \ \cup\ 
         \bigl\{\ket{w_aw_0x_0}\mid a\in\mathbb F_q^\times\bigr\}\; .
    \end{aligned}
    \end{equation}
\end{theorem}

\begin{proof}
    Denote the four blocks in Eq.~(\ref{eq: unextendible qudit stabiliser set}) by $\mathcal B_0,\mathcal B_1,\mathcal B_2,\mathcal B_3$, respectively. Clearly, each block is orthogonal internally, while between pairs
    \begin{align*}
        (\mathcal B_0,\mathcal B_1),\
        (\mathcal B_0,\mathcal B_2),\
        (\mathcal B_0,\mathcal B_3),\
        (\mathcal B_1,\mathcal B_2),\
        (\mathcal B_1,\mathcal B_3),\
        (\mathcal B_2,\mathcal B_3),
    \end{align*}
    orthogonality occurs on qudits
    $2,1,3,3,1,2$, respectively, hence, $\cS_q$ is pairwise orthogonal. Note also that $|\cS_q|=(q-1)^3+(q-1)^2+(q-1)+1=q^3-2q(q-1)<q^3$, hence, $\cS_q$ is not a stabiliser basis. 
    
    Suppose that $\ket\psi$ is a stabiliser state orthogonal to every element of $\cS_q$. By Lm.~\ref{lm: affine stabiliser support}, its outcome support $S$ for the Abelian subgroup $\langle X_1, Z_2, X_3\rangle$ is an affine subspace of $\mathbb F_q^3$. Orthogonality to $\mathcal B_0$ implies $S\cap(\mathbb F_q^\times)^3=\emptyset$, and therefore
    \begin{align*}
        S\subseteq H_1\cup H_2\cup H_3\; ,\qquad
        H_i:=\{s_i=0\}\; .
    \end{align*}
    If $S$ is contained in none of the $H_i$, then every
    $S\cap H_i$ has at most $|S|/q$ elements. Hence,
    \begin{align}\label{eq: d>3 obstruction}
        |S|
        \leq\sum_{i=1}^3|S\cap H_i|
        \leq\frac{3}{q}|S|
        <|S|\; ,
    \end{align}
    a contradiction. Thus $S\subseteq H_i$ for some $i$, and $\ket\psi$ has one of the deterministic local stabilisers $X_1=0$, $Z_2=0$ or $X_3=0$. Accordingly, $\ket\psi$ factorises across the corresponding qudit. After deleting the fixed factor, the remaining orthogonality conditions are precisely those determined by $\cT_1$, $\cT_2$ and $\cT_3$ in Lm.~\ref{lm: two-qudit obstructions}, hence, yield a contradiction.
\end{proof}

We are left to treat the qutrit case, for which the obstruction in Eq.~(\ref{eq: d>3 obstruction}) does not hold. Still, the construction in Thm.~\ref{thm: unextendible qudits, d>3} can be adapted to cover the case $q=3$ as well. To this end, we will need the following lemma.

\begin{lemma}\label{lm: affine cover}
    Let $S\subseteq\mathbb F_3^3$ be a nonempty affine subspace, that is, $S=s+U$ for $U\subset \mathbb F_3^3$ linear, satisfying
    \begin{align*}
        S\subseteq H_1\cup H_2\cup H_3\; ,\qquad
        H_i=\{s_i=0\}\; .
    \end{align*}
    Then either $S\subseteq H_i$ for some $i$, or $S=s+\mathbb F_3r$ is an affine line whose direction $r=(r_1,r_2,r_3)$ satisfies $r_1r_2r_3\neq0$. In the latter case, $S$ meets each $H_i$ in exactly one point.
\end{lemma}

\begin{proof}
    Suppose that $S\not\subset H_i$ for $i\in\{1,2,3\}$, and let $S=s+U$ for $U\subset\mathbb F_3^3$ a linear subspace of dimension $m=\dim U=\dim S$. (Note that our assumption implies $m>0$ since otherwise $S=\{s\}$ and thus $S\subset H_i$ for some $i\in\{1,2,3\}$.) As an intersection of affine hyperplanes, each set $S\cap H_i$ contains at most $3^{m-1}$ elements, hence, the inequality,
    \begin{align*}
        3^m=|S|
        \leq\sum_{i=1}^3|S\cap H_i|
        \leq3\cdot3^{m-1}\; ,
    \end{align*}
    is in fact an equality, and the three intersections therefore define pairwise disjoint affine hyperplanes of $S$.
    
    If $m\geq2$, disjointness implies that these hyperplanes are parallel, that is, the linear spaces $U_i$ in $S\cap H_i=h_i+U_i$ for $i\in\{1,2,3\}$ all coincide. Consequently, the three coordinate functionals $l_i(x_1,x_2,x_3)=x_i$ restricted to $U$ all have the same kernel, and since two nonzero linear functionals on a vector space with the same kernel are proportional, there exist $c_1,c_2,c_3\in\mathbb F_3^\times$ and a nonzero functional $f\in U^*$ such that $l_i|_U=c_if$.
    
    It follows that the coordinate map $\iota:U\longrightarrow\mathbb F_3^3$ takes the form
    \begin{align*}
        u\longmapsto
        \bigl(\ell_1(u),\ell_2(u),\ell_3(u)\bigr)
        =f(u)(c_1,c_2,c_3)\; .
    \end{align*}
    that is, $\mathrm{im}(\iota)\subseteq\langle(c_1,c_2,c_3)\rangle$, hence, $\mathrm{rank}(\iota)\leq1$. On the other hand, $\iota$ is simply the inclusion of the subspace $U\subset\mathbb F_3^3$, written in coordinates. It is therefore injective, that is, $\mathrm{rank}(\iota)=\dim U=m$. We thus conclude $m<2$, contradicting $m\geq2$. Consequently, we must have $m=1$, that is, $U=\langle r\rangle$ for some nonzero vector $r=(r_1,r_2,r_3)\in\mathbb F_3^3$, and thus $S=s+\mathbb F_3r$. Moreover, since $S\nsubseteq H_i$ we have $l_i|_U\neq0$, that is, $r_i\neq0$ for all $i\in\{1,2,3\}$, equivalently $r_1r_2r_3\neq0$.
    
    Finally, since for every $i\in\{1,2,3\}$ the linear equation $s_i+t r_i=0$ has a unique solution $t\in\mathbb F_3$, $S$ meets each coordinate hyperplane $H_i$ in a single point.
\end{proof}

\begin{theorem}\label{thm: unextendible qutrits}
    The set $\cS_3$ obtained from Eq.~(\ref{eq: unextendible qudit stabiliser set}) for $q=3$ is an unextendible stabiliser basis.
\end{theorem}

\begin{proof}
    Pairwise orthogonality holds as for $q\geq 5$ in Thm.~\ref{thm: unextendible qudits, d>3}, and $\cS_3$ is not a stabiliser basis since $|\cS_3|=15<27=3^3$.
    
    Suppose now that a stabiliser state $\ket{\psi}$ corresponding to the affine Lagrangian $A=a+L\subseteq V=\mathbb F_3^6$ is orthogonal to every element of $\cS_3$, let $E=\langle e_1,e_2,e_3\rangle$ be the Lagrangian corresponding to the Abelian subgroup $\langle X_1,Z_2,X_3\rangle$. Moreover, define the map $\pi:V\longrightarrow\mathbb F_3^3$ by $\pi(v)=\bigl(\omega(v,e_1),\omega(v,e_2),\omega(v,e_3)\bigr)$. By Lm.~\ref{lm: affine stabiliser support}, the outcome support of the state $\ket{\psi}$ is the affine subspace $S:=\pi(A)$. As in the proof of Thm.~\ref{thm: unextendible qudits, d>3}, orthogonality to the first block of $\cS_3$ implies $S\subseteq H_1\cup H_2\cup H_3$. If $S\subseteq H_i$, the candidate has one of the deterministic local stabilisers $X_1=0$, $Z_2=0$ or $X_3=0$, and Lm.~\ref{lm: two-qudit obstructions} results in a contradiction.
    
    It remains to consider the exceptional case of Lm.~\ref{lm: affine cover}, that is, $S=s+\mathbb F_3r$ with $r_1r_2r_3\neq0$. Since $\pi(L)=\langle r\rangle$ in this case, the kernel of $\pi|_L$ is the two-dimensional space $K:=L\cap E$. Isotropy of $L$ gives
    \begin{align*}
        K
        =\left\{\ \sum_{i=1}^3u_ie_i\ \bigg\mid\ u_1r_1+u_2r_2+u_3r_3=0\ \right\}\; .
    \end{align*}
    By Lm.~\ref{lm: affine cover}, $S$ intersects $H_2$, hence, there exists $p\in A$ with $\omega(p,e_2)=0$. Consider the stabiliser state $\ket{w_0z_0w_0}$, whose associated Lagrangian is given by $L_{\ket{w_0z_0w_0}}=\{v\in V\mid\omega(v,f_1)=\omega(v,e_2)=\omega(v,f_3)=0\}$, where $f_1$ and $f_3$ are the representatives of the Pauli operators $W_1$ and $W_3$, in particular, $\omega(e_1,f_1)=1=\omega(e_3,f_3)$. Now, choose
    \begin{align*}
        u_1&=-\omega(p,f_1)\; ,&
        u_3&=-\omega(p,f_3)\; ,
    \end{align*}
    and note that, since $r_2\neq0$, there exists a unique element $u_2$ such that $u_1r_1+u_2r_2+u_3r_3=0$. By construction,
    \begin{align*}
        k
        :=u_1e_1+u_2e_2+u_3e_3
    \end{align*}
    belongs to $K\subseteq L$, and $p':=p+k$ satisfies all three equations in $L_{\ket{w_0z_0w_0}}$, hence, $p'\in A\cap L_{\ket{w_0z_0w_0}}$. By Cor.~\ref{cor: affine stabiliser overlap}, this implies that $\psi$ has nonzero overlap with $\ket{w_0z_0w_0}$, contradicting the assumption of orthogonality.
\end{proof}

\begin{proof}[Proof of Thm.~\ref{thm: 3 qudit USB}]
    The result follows from Thm.~\ref{thm: unextendible qutrits} and Thm.~\ref{thm: unextendible qudits, d>3}.
\end{proof}

\section{Proof of Thm.~\ref{thm: 3 qubits symplectic}}\label{app: proof - 3 qubits}

We consider the untwisted qubit stabiliser theory over $V=\zz^{2n}_2$. More precisely, in analogy with Eq.~(\ref{eq: stabiliser projector}), we define
\begin{align}
    \pi_{L,\xi}
    =\frac{1}{2^n}\sum_{v\in L}(-1)^{\xi(v)}e_v\; ,
\end{align}
to be a vector in the Bloch space $\R[V]$ with orthogonal basis $\{e_v\}_{v\in V}$ with $\langle e_v,e_w\rangle=2^n\delta_{vw}$. The orthogonality relations for two vectors $\pi_{L,\xi},\pi_{L',\xi'}$ with $L,L'\in\Lag(V)$ and $\xi\in L^*,\xi'\in L'^*$ are then easily seen to be given by
\begin{align}\label{eq: stabiliser orthogonality}
    \langle\pi_{L,\xi},\pi_{L',\xi'}\rangle=0
    \quad\Longleftrightarrow\quad
    \xi_{L\cap L'}\neq\xi'_{L\cap L'}\; ,
\end{align}
equivalently by Eq.~\ref{eq: affine stabiliser overlap}. Within the space of all such abstract stabiliser states $\cS^\mathrm{symp}_n:=\{\pi_{L,\xi}\mid L\in\Lag(V),\ \xi\in L^*\}$, we ask whether there exist sets of pairwise orthogonal vectors that cannot be extended to an abstract stabiliser basis of $2^n$ orthogonal such vectors in $\cS^\mathrm{symp}_n$. The answer is again positive, yet unlike its twisted qubit cousin, already obtains at $n=3$. To construct such a set, we will use the following well-known construction from finite projective geometry.\\

\textbf{Klein correspondence.} We recall the basics of a well-known finite-geometric construction (for more details, see Refs.~\cite[Chs.~7, 8, and 12]{Taylor1992} and \cite[Chs.~2 and 5]{Hirschfeld1998}). For $q$ prime, the projective space $\PG(n,q)$ is the projectivisation of $\mathbb{F}_q^{n+1}$. Its points are the one-dimensional subspaces $\langle x\rangle$, where $0\neq x\in\mathbb{F}_q^{n+1}$. When $q=2$, every one-dimensional subspace contains a unique nonzero vector, so points $[x]\in\PG(n,2)$ may be identified with nonzero vectors $0\neq x\in\mathbb{F}_2^{n+1}$.

Let $E\cong\mathbb{F}_q^{2r}$ carry a symplectic (nondegenerate alternating bilinear) form $\omega$. The symplectic polar space $W(2r-1,q)$ consists of the projectivisations of the totally isotropic subspaces $U\subset E$, that is, $\omega_U=\omega|_{U\times U}=0$. In particular, the nonzero isotropic subspaces of the $n$-qubit symplectic vector space $\mathbb{F}_2^{2n}$ correspond to the projective subspaces of $W(2n-1,2)$. A Lagrangian vector subspace has dimension $n$ and hence corresponds to a projective $(n-1)$-subspace.

Let $W$ be a four-dimensional vector space over $\mathbb{F}_2$, and define the vector space of bi-vectors over it by
\begin{align*}
    V:=\bigwedge\nolimits^2 W\; .
\end{align*}
Since $\dim V=\binom{4}{2}=6$, $V$ is isomorphic as a vector space to the three-qubit vector space $\mathbb{F}_2^6$. Choose a nonzero volume element $\Omega\in\bigwedge^4W$ and define a bilinear form $\omega_K:V\times V\to\mathbb{F}_2$ by
\begin{align}\label{eq: symplectic form}
    z\wedge z'=\omega_K(z,z')\Omega\; .
\end{align}
The form $\omega_K$ is nondegenerate and alternating, hence, $(V,\omega_K)$ is a six-dimensional symplectic vector space.

Let $e_0,e_1,e_2,e_3$ be a basis of $W$, write $V\ni z=\sum_{0\leq i<j\leq 3}p_{ij}e_i\wedge e_j$, and define the Klein quadratic form on $V$ by
\begin{align*}
    Q_K(z)
    :=p_{01}p_{23}+p_{02}p_{13}+p_{03}p_{12}\; .
\end{align*}
Its polar form is $\omega_K$, and since $Q_K$ is the orthogonal direct sum of three hyperbolic planes (corresponding to the coordinate pairs $(p_{01},p_{23})$, $(p_{02},p_{13})$ and $(p_{03},p_{12})$), it has Witt index three and thus defines the hyperbolic quadric
\begin{align*}
    Q^+(5,2)
    :=\big\{[z]\in\PG(V)\mid Q_K(z)=0\big\}\; .
\end{align*}
A nonzero bi-vector $z\in\bigwedge^2W$ is decomposable if $z=x\wedge y$ for some linearly independent $x,y\in W$. In dimension four, decomposability is equivalent to the Pl\"ucker relation $Q_K(z)=0$. Consequently, the Klein correspondence
\begin{align}\label{eq: Klein correspondence}
    \kappa:\Gr(2,W)
    \longrightarrow Q^+(5,2)\; ,\qquad
    \langle x,y\rangle\longmapsto[x\wedge y],
\end{align}
is a bijection between the Grassmannian of all projective lines of $\PG(W)$ and the points of the Klein quadric.

Note that for two nonzero decomposable bi-vectors $z=x\wedge y$ and $z'=x'\wedge y'$, one has
\begin{align*}
    z\wedge z'\neq 0
    \quad\Longleftrightarrow\quad
    \langle x,y\rangle\cap\langle x',y'\rangle={0}.
\end{align*}
Thus their wedge product is nonzero precisely when the corresponding projective lines of $\PG(3,2)$ are skew, that is, have no projective point in common (equivalently, the corresponding two-dimensional subspaces intersect trivially).

The maximal projective subspaces contained in $Q^+(5,2)$ (equivalently, Lagrangian subspaces $L\subset V$ with $L\backslash\{0\}\subset Q^+(5,2)$) are projective planes and occur in two families: first, for every $0\neq x\in W$, the plane
\begin{align*}
    \alpha_x
    :=\mathbb{P}(x\wedge W).
\end{align*}
represents all projective lines through the point $[x]$; second, for every hyperplane $H<W$, the plane
\begin{align*}
    \beta_H
    :=\mathbb{P}\left(\bigwedge\nolimits^2H\right)\; .
\end{align*}
represents all projective lines contained in the projective plane $\mathbb{P}(H)$. Indeed, the underlying three-dimensional vector spaces $x\wedge W$ and $\bigwedge^2H$ are totally singular with respect to $Q_K$ and therefore Lagrangian with respect to $\omega_K$.

Finally, let $p=[x]\in\mathbb{P}(H)$. The \emph{pencil with vertex $p$} in the plane $\mathbb{P}(H)$ is defined as
\begin{align*}
    \mathcal{P}(p,\mathbb{P}(H))
    :=\left\{l\mid p\in l\subseteq\mathbb{P}(H)\right\}\; .
\end{align*}
Under the Klein correspondence in Eq.~(\ref{eq: Klein correspondence}), this pencil is represented by the projective line $\alpha_x\cap\beta_H\subseteq Q^+(5,2)$. Conversely, every projective line contained in the Klein quadric arises from an incident point--plane pair in this way.

The relations between the respective objects under the Klein correspondence are summarised in Tab.~\ref{tab: Klein correspondence}. Using them, we now prove the existence of USBs for the (untwisted) symplectic theory underlying the $n$-qubit Pauli group.

\begin{table}[t]
\centering
\begin{tabular}{lll}
    \toprule
    &
    $\PG(3,2)=\mathbb P(W)$
    &
    $Q^+(5,2)\subseteq\mathbb P(\Lambda^2W)$
    \\
    \midrule
    
    \multirow{4}{*}{\textbf{configurations}}
    &
    projective line $\ell=\mathbb P\langle x,y\rangle$
    &
    point $\kappa(\ell)=[x\wedge y]$
    \\
    
    &
    pencil of lines through $p=[x]$ and contained in $\pi=\mathbb P(H)$ 
    &
    projective line $\mathbb P(x\wedge H)=\alpha_x\cap\beta_H$
    \\
    
    &
    projective lines through a fixed point $p=[x]$
    &
    generator plane $\alpha_x:=\mathbb P(x\wedge W)$
    \\
    
    &
    all projective lines contained in a fixed plane $\pi=\mathbb P(H)$
    &
    generator plane $\beta_H:=\mathbb P(\Lambda^2H)$
    \\[1mm]
    
    \midrule
    
    \multirow{4}{*}{\textbf{incidences}}
    &
    line
    $\mathbb P\langle x,y\rangle$ joining
    distinct points $[x]$, $[y]$
    &
    $\alpha_x\cap\alpha_y=\{[x\wedge y]\}=\{\kappa(\mathbb P\langle x,y\rangle)\}$
    \\
    
    &
    line $\mathbb P(H\cap H')$ intersecting distinct planes $\mathbb P(H)$, $\mathbb P(H')$
    &
    $\beta_H\cap\beta_{H'}=\mathbb P\!\left(\Lambda^2(H\cap H')\right)=\{\kappa(\mathbb P(H\cap H'))\}$
    \\
    
    &
    pencil through $p=[x]$ and contained in
    $\pi=\mathbb P(H)$ 
    &
    $\alpha_x\cap\beta_H=\mathbb P(x\wedge H)$,
    a projective line
    \\
    
    &
    no incident point--plane pair, when $x\notin H$
    &
    $\alpha_x\cap\beta_H=\emptyset$
    \\
    \bottomrule
\end{tabular}
\caption{The Klein correspondence.}
\label{tab: Klein correspondence}
\end{table}

\begin{proof}[Proof of Thm.~\ref{thm: 3 qubits symplectic}]
    Consider the four-dimensional vector space
    \begin{align*}
        W
        :=\left\{(t_0,\ldots,t_4)\in\mathbb{F}_2^5\ \bigg\mid\ \sum_{i=0}^4t_i=0\right\}\; ,
    \end{align*}
    and define $x_i:=\mathbf{1}+e_i$ for all $i\in\mathbb{Z}_5$, where $\mathbf{1}=(1,1,1,1,1)$ and $e_i$ is the vector with only nonzero value at the $i$-th position. Note that $\sum_{i=0}^4x_i=0$, and this is the only nontrivial linear relation between these vectors, in particular, any four of the $x_i$ form a basis of $W$. With respect to the basis $x_0,x_1,x_2,x_3$, define the quadratic form
    \begin{align*}
        \fq(a_0,a_1,a_2,a_3)
        :=\sum_{0\leq r<s\leq3}a_ra_s\; .
    \end{align*}
    Note that $\fq(a)=\binom{m}{2}\pmod 2$ for the binary vector $a=(a_0,a_1,a_2,a_3)$ with Hamming weight $m$ (where $\binom{0}{2}=\binom{1}{2}=0$). Consequently, $\fq^{-1}(0)\backslash\{0\}=\{x_0,x_1,x_2,x_3,x_4\}$, where $x_4=x_0+x_1+x_2+x_3$. The polar form $b_\fq(a,b)=\sum_{0\leq r<s\leq 3}(a_rb_s+a_sb_r)$ of $\fq$ has matrix representation $b_\fq=\begin{pmatrix}
        0 & 1 & 1 & 1 \\ 1 & 0 & 1 & 1 \\ 1 & 1 & 0 & 1 \\ 1 & 1 & 1 & 0
    \end{pmatrix}$ from which it immediately follows that it is alternating and nondegenerate. Moreover, since the above zero set contains no projective line, $\fq$ has Witt index one, and ${[x_i]}_{i\in\mathbb{Z}_5}$ thus corresponds with the elliptic quadric $Q^-(3,2)\subseteq\PG(3,2)$, equivalently the five-point ovoid of $\PG(3,2)$.\footnote{This ovoid corresponds with maximal sets (with cardinality $5$) of mutually anti-commuting Pauli operators.}

    Next, we identify ${[x_i]}_{i\in\mathbb{Z}_5}$ with Lagrangians in $\mathbb F^6_2$ via the Klein correspondence. Define $V:=\bigwedge\nolimits^2W$ and equip $V$ with the symplectic form $\omega_K$ defined in Eq.~(\ref{eq: symplectic form}) above. For every $i\in\mathbb{Z}_5$, let $L_i:=x_i\wedge W$ and define the map $W\ra L_i$ by $y\mapsto x_i\wedge y$. Its kernel is $\langle x_i\rangle$, and thus $\dim(L_i)=3$. Moreover, every element of $L_i$ is decomposable and thus $Q_K|_{L_i}=0$, hence, $L_i$ is totally singular with respect to $Q_K$ and defines a Lagrangian subspace with respect to $\omega_K$.
    
    Next, we note that for distinct $i,j\in\mathbb{Z}_5$, one has $L_i\cap L_j=\langle x_i\wedge x_j\rangle$. We orient the complete graph on $\mathbb{Z}_5$ cyclically, write $i\longrightarrow j$ if and only if $j-i\in{1,2}\pmod 5$, and for every $i\in\mathbb{Z}_5$ define the linear functional $\ell_i\in W^*$ by 
    \begin{align}\label{eq: cyclic orientation}
        \ell_i(x_j)
        \ :=\ \begin{cases}
            0 & j=i\\
            0 & i\longrightarrow j\\
            1 & j\longrightarrow i
        \end{cases}\; .
    \end{align}
    For fixed $i$, exactly two of the values $\ell_i(x_j)$ equal to one. The assignment thus respects the binary relation $\sum_jx_j=0$, and extends uniquely to a linear functional on $W$. We define the linear functionals $\xi_i\in L_i^*$ by
    \begin{align}\label{eq: characters}
        \xi_i(x_i\wedge y)
        :=\ell_i(y)\; ,
    \end{align}
    which are well-defined because the representative $y$ is defined modulo $\langle x_i\rangle$ and $\ell_i(x_i)=0$.
    
    Now, consider the five symplectic stabiliser states $\pi_i:=\pi_{L_i,\xi_i}$, for $i\in\mathbb{Z}_5$. For $i\neq j$, Eq.~(\ref{eq: characters}) gives $\xi_i(x_i\wedge x_j)=\ell_i(x_j)$ and $\xi_j(x_i\wedge x_j)=\ell_j(x_i)$, and since either $i\longrightarrow j$ or $j\longrightarrow i$ (but not both), we have $\ell_i(x_j)\neq\ell_j(x_i)$. Consequently, the functiona=ls $\{\xi_i\}_{i\in\zz_5}$ disagree on overlap $L_i\cap L_j$, hence, the vectors $\{\pi_i\}_{i\in\mathbb{Z}_5}$ are pairwise orthogonal.\\
    
    It remains to prove that this set of pairwise orthogonal symplectic stabiliser states is unextendible. Assume to the contrary that another symplectic stabiliser state $\pi_{M,\eta}$ is orthogonal to every $\pi_i$, in particular, $M\cap L_i\neq\{0\}$ for every $i\in\mathbb{Z}_5$ (by Eq.~(\ref{eq: stabiliser overlap})). Since $M$ is Lagrangian for the polar form of $Q_K$, the restriction $Q_K|_M$ is linear.
    
    Suppose first that $Q_K|_M\neq0$. Its kernel then has vector dimension two, hence, $\mathbb{P}(M)\cap Q^+(5,2)$ is a projective line on the Klein quadric, which represents a pencil of projective lines through some point $p$ and contained in some projective plane $\mathbb{P}(H)$ for a hyperplane $H<W$ (see Tab.~\ref{tab: Klein correspondence}). Since $Q_K|_{L_i}=0$, every nonzero element of $M\cap L_i$ lies on the Klein quadric. The condition $M\cap L_i\neq\{0\}$ therefore implies that the Klein line representing the pencil meets every generator plane $\mathbb{P}(L_i)=\alpha_{x_i}$. Yet, a pencil contained in $\mathbb{P}(H)$ meets $\alpha_{x_i}$ only if $[x_i]\in\mathbb{P}(H)$ (see Tab.~\ref{tab: Klein correspondence}), which therefore implies $x_i\in H$ for every $i$, contradicting $\mathrm{span}\{x_0,\ldots,x_4\}=W$. Consequently, we must have $Q_K|_M=0$.
    
    If $Q_K|_M=0$ then $\mathbb{P}(M)$ is a generator plane of the Klein quadric, which is of one of the two forms (see Tab.~\ref{tab: Klein correspondence})
    \begin{align*}
        \mathbb{P}(x\wedge W)
        \qquad\text{or}\qquad
        \mathbb{P}\left(\bigwedge\nolimits^2H\right)\; ,
    \end{align*}
    where $0\neq x\in W$ and $H<W$ is a hyperplane. A generator of the second type meets $\alpha_{x_i}$ only if $x_i\in H$, so it cannot meet all five $\alpha_{x_i}$. Therefore, we must have $M=L_x:=x\wedge W$ for some $0\neq x\in W$.
    
    Suppose first that $x=x_k$ for some $k\in\mathbb{Z}_5$. For every $i\neq k$, orthogonality of $\pi_{M,\eta}$ to $\pi_i$ then requires
    \begin{align*}
        \eta(x_k\wedge x_i)
        =1+\xi_i(x_k\wedge x_i)
        =1+\ell_i(x_k)\; .
    \end{align*}
    By Eq.~(\ref{eq: cyclic orientation}), $1+\ell_i(x_k)=\ell_k(x_i)$, hence, $\eta(x_k\wedge x_i)=\xi_k(x_k\wedge x_i)$ for every $i\neq k$. The vectors $x_k\wedge x_i$, with $i\neq k$, span $L_k$, so $\eta=\xi_k$, which implies $\pi_{M,\eta}=\pi_k$. Yet, since $\pi_{M,\eta}$ is clearly not orthogonal to itself, this is a contradiction.
    
    This leaves $x\notin{x_0,\ldots,x_4}$. Orthogonality to the $\pi_i$ requires $\eta(x\wedge x_i)=1+\ell_i(x)$ for every $i\in\mathbb{Z}_5$. On the other hand,
    \begin{align*}
        \sum_{i=0}^4x\wedge x_i
        =x\wedge\sum_{i=0}^4x_i\
        =0\; ,
    \end{align*}
    and by linearity of $\eta$, the preceding orthogonality conditions
    would imply
    \begin{align}\label{eq: =0}
        \sum_{i=0}^4\bigl(1+\ell_i(x)\bigr)=0\; .
    \end{align}
    However, since for every vector $x_j$, exactly two indices $i$ satisfy $\ell_i(x_j)=1$ (by Eq.~(\ref{eq: cyclic orientation})), we have $\sum_{i=0}^4\ell_i=0$ and thus
    \begin{align*}
        \sum_{i=0}^4\bigl(1+\ell_i(x)\bigr)
        =\sum_{i=0}^4 1+\left(\sum_{i=0}^4\ell_i\right)(x)\
        =1\; ,
    \end{align*}
    contradicting Eq.~(\ref{eq: =0}). Consequently, there exists no further symplectic stabiliser state orthogonal to all five $\pi_i$, and $\{\pi_{L_i,\xi_i}\}_{i\in\mathbb{Z}_5}$ is thus a maximal set of pairwise orthogonal three-qubit symplectic stabiliser states.
\end{proof}

\section{Minimality and lifting property}\label{app: minimality and lifting}

We prove two general extendibility results: in Sec.~\ref{app: 2-extendibility}, we show that sets of pairwise orthogonal two-qudit stabiliser states of prime dimension and, in Sec.~\ref{app: 3-extendibility}, that pairwise orthogonal sets of three-qubit stabiliser states can always be extended to a stabiliser basis. Finally, we show that USBs lift to USBs involving more parties in Sec.~\ref{app: lifting}.

\subsection{Sets of pairwise orthogonal two-qudit
stabiliser states extend to a stabiliser basis}\label{app: 2-extendibility}

We begin with the elementary polar-space fact (underlying both the
signed and unsigned two-qubit statements).

\begin{lemma}\label{lm: rank-two-pencil}
    Let $V=\mathbb F^4$ be a symplectic vector space over a field $\mathbb F$. If a family of Lagrangians $\{L_i\}_{i\in I}$ satisfies
    \begin{align*}
        L_i\cap L_j\neq\{0\}
        \qquad\forall i,j\in I\; ,
    \end{align*}
    then there is a one-dimensional subspace $D\leq V$ such that $D\subseteq L_i$ for all $i\in I$.
\end{lemma}

\begin{proof}
    If all $L_i$ coincide, the assertion is immediate. Otherwise, let $D:=L_1\cap L_2$ for distinct $L_1,L_2$. Then $\dim D=1$ and $L_1+L_2=D^\perp$. Let $L$ be another member of the set, assume that $D\nsubseteq L$, and pick nonzero $u\in L\cap L_1$ and $v\in L\cap L_2$. Then $L=\langle u,v\rangle\subseteq D^\perp$, and since $D\nsubseteq L$, the image of $L$ in the two-dimensional nondegenerate symplectic space $D^\perp/D$ is two-dimensional and isotropic. Clearly, this is impossible, and we must conclude $D\subseteq L$.
\end{proof}

\begin{theorem}\label{thm: 2 qudit extendibility}
    Every set of pairwise orthogonal two-qudit stabiliser states (for $q$ prime) extends to a stabiliser basis.
\end{theorem}

\begin{proof}
    Let $L_i\leq\mathbb F_q^4$ be the Lagrangians of the stabiliser states in the set. The overlap criterion in Lm.~\ref{lm: stabiliser overlap} implies that $L_i\cap L_j\neq\{0\}$ for all $i\neq j$. By Lm.~\ref{lm: rank-two-pencil}, the $L_i$ thus contain a common one-dimensional subspace $D$.

    We can thus choose a Clifford unitary which sends the common Pauli operator corresponding to the subspace $D$ to the Pauli operator $Z_1$. After this transformation, every stabiliser state in the set has the form
    \begin{align}\label{eq: prodduct form}
        \ket{\psi_i}
        =\ket{z_{t_i}}_1\otimes\ket{\phi_i}_2\; ,
    \end{align}
    for $t_i\in\mathbb F_q$ and where $\ket{\phi_i}$ is a one-qudit stabiliser state. Next, partition the family according to the value of $t_i$. States in different parts are automatically orthogonal on the first qudit, while for fixed $t_i=t$, orthogonality is equivalent to orthogonality of the residual one-qudit states (see Eq.~(\ref{eq: stabiliser overlap})), that is,
    \begin{align*}
        \bracket{\psi_i}{\psi_j}=0
        \quad\Longleftrightarrow\quad
        \bracket{\phi_i}{\phi_j}=0.
    \end{align*}
    For a qubit, any two orthogonal pure stabiliser states are the opposite eigenstates of one Pauli operator, hence, a pairwise orthogonal set lies in a single stabiliser basis. For a qudit of odd prime dimension, stabiliser states are represented by affine lines in $\mathbb F_q^2$, and two such lines are disjoint if and only if they are parallel. Consequently, every set of pairwise orthogonal one-qudit stabiliser states belongs to a single one-qudit stabiliser basis. For each $t\in\mathbb F_q$, we can therefore choose a one-qubit, respectively one-qudit stabiliser basis $\cB_t$ containing all residual states occurring in the $t$-th part (if no state has label $t$, we may choose $\cB_t$ arbitrarily). Then
    \begin{align*}
        \cB
        \ :=\ \bigcup_{t\in\mathbb F_q}\left\{\ket{z_t}\otimes\ket{\beta}\mid\ket{\beta}\in\mathcal B_t\right\}\; ,
    \end{align*}
    is an orthonormal basis of $q^2$ two-qudit stabiliser states, that is, a stabiliser basis containing the states in Eq.~(\ref{eq: prodduct form}). Finally, undoing the Clifford unitary proves the result.
\end{proof}

\subsection{Sets of pairwise orthogonal three-qubit stabiliser states extend to a stabiliser basis}\label{app: 3-extendibility}

The fact that every set of pairwise orthogonal three-qubit states extends to a stabiliser basis can be readily checked by an exhaustive computer search. The following analytical argument employs a beautiful fact of finite symplectic geometry, discussed in more detail in Ref.~\cite{LevayPlanatSangia2013,VanGeemenMarrani2019}, and exposes the distinction between the qubit and qudit case.\\

\textbf{Spin embedding.} Let $V:=\mathbb F_2^6$ be the three-qubit symplectic space, and let $\Lag(V)$ denote its $135=\prod_{i=1}^3(2^i+1)$ Lagrangian subspaces. We will use the `spin embedding' from Ref.~\cite{LevayPlanatSangia2013,VanGeemenMarrani2019},
\begin{align}\label{eq: spin embedding}
    \sigma:\Lag(V)
    \longrightarrow Q^+(7,2)\subseteq\mathrm{PG}(7,2)\; ,
\end{align}
which is a bijection from the Lagrangians of $V$ to the points of the hyperbolic quadric $Q^+(7,2)$, satisfying
\begin{align}\label{eq: spin orthogonality}
    L\cap M\neq\{0\}
    \quad\Longleftrightarrow\quad
    \sigma(L)\perp\sigma(M)\; .
\end{align}

For completeness, we add the following lemma proving these relations.

\begin{lemma}\label{lm: spin embedding}
    Let $\mathbb F^6_2\cong V=X\oplus Z$ be a six-dimensional symplectic vector space, where $X$ and $Z$ are complementary Lagrangian subspaces, and write $\omega\bigl((x,z),(y,w)\bigr)=x^Tw+z^Ty$.
    
    \begin{enumerate}
        \item[(i)] A $Z$-transverse Lagrangian $L$, that is, $L\cap Z=\{0\}$, is given by $L_A=\{(x,Ax)\mid x\in X\}$, where $A$ is a symmetric $3\times3$ matrix. Indeed, $\omega\bigl((x,Ax),(y,Ay)\bigr)=x^T(A+A^T)y$. In particular, $L_A$ is isotropic if and only if $A=A^T$.
        
        \item[(ii)] For $I\subseteq\{1,2,3\}$ let $p_I(A)$ be the corresponding principal minor of $A$, with $p_\emptyset(A):=1$, and define
        \begin{align*}
            \sigma(L_A)
            =[p_\emptyset:p_1:p_2:p_3:p_{23}:p_{13}:p_{12}:p_{123}]\in\mathrm{PG}(7,2)\; .
        \end{align*}
        These coordinates satisfy $Q_{\mathrm{spin}}(\sigma(L_A))=0$, where $Q_{\mathrm{spin}}$ is the hyperbolic quadric $Q^+(7,2)$, given by
        \begin{align*}
            Q_{\mathrm{spin}}(p)
            =p_\emptyset p_{123}
            +p_1p_{23}+p_2p_{13}+p_3p_{12}\; .
        \end{align*}
        Its polar form is $B_{\mathrm{spin}}(p,p')=p_\emptyset p'_{123}+p_{123}p'_\emptyset+p_1p'_{23}+p_{23}p'_1+p_2p'_{13}+p_{13}p'_2+p_3p'_{12}+p_{12}p'_3$, and for two $Z$-transverse Lagrangians $L_A,L_B$, it holds  $B_{\mathrm{spin}}\bigl(\sigma(L_A),\sigma(L_B)\bigr)=\det(A+B)$, consequently,
        \begin{align}\label{eq: transverse orthogonality}
            L_A\cap L_B\neq\{0\}
            \quad\Longleftrightarrow\quad
            \sigma(L_A)\perp\sigma(L_B)\; .
        \end{align}
        
        \item[(iii)] The map $\sigma$ extends from the symmetric-matrix charts in (i) and (ii) to an $\mathrm{Sp}(6,2)$-equivariant bijection
        \begin{align*}
            \sigma:\Lag(V)\overset{\sim}{\longrightarrow}Q^+(7,2)\; .
        \end{align*}
        It is called the \emph{spin embedding} of the binary Lagrangian Grassmannian \cite{LevayPlanatSangia2013,VanGeemenMarrani2019}, and satisfies
        \begin{align*}
            L\cap M\neq\{0\}
            \quad\Longleftrightarrow\quad
            \sigma(L)\perp\sigma(M)
            \quad\quad\forall L,M\in\Lag(V)\; .
        \end{align*}
    \end{enumerate}
\end{lemma}

\begin{proof}
    \textbf{(i).} Fix complementary Lagrangians $X,Z\in\Lag(V)$ such that $V=X\oplus Z$ and identify $Z$ with $X^*$ by means of the symplectic form $\omega$. If $L$ is transverse to $Z$, then the projection map $V=X\oplus Z\ra X$ identifies $L$ with the graph $L_A=\{(x,Ax)\mid x\in X\}$ of a unique linear map $A:X\to Z$. For $x,y\in X$, we have
    \begin{align*}
        \omega((x,Ax),(y,Ay))
        =x^T(A+A^T)y\; .
    \end{align*}
    Hence, $L_A$ is Lagrangian if and only if
    $A=A^T$. Here, $A$ is the standard symmetric-matrix chart of $\Lag(V)$, viewed as a projective variety (see
    \cite[\S4.2]{VanGeemenMarrani2019}).
    
    \textbf{(ii).} Let $A=\begin{pmatrix}
    a&u&v\\ u&b&w\\ v&w&c
    \end{pmatrix}$, its principal minors are $p_{123}=\det A=abc+aw^2+bv^2+cu^2$, $p_\emptyset=1$ (by definition) and 
    \begin{align*}
        p_1=a\; ,\qquad
        p_2=b\; ,\qquad
        p_3=c\; ,\qquad
        p_{12}=ab+u^2\; ,\qquad
        p_{13}=ac+v^2\; ,\qquad
        p_{23}=bc+w^2\; .
    \end{align*}
    One immediately checks that
    \begin{align*}
        p_\emptyset p_{123}+p_1p_{23}+p_2p_{13}+p_3p_{12}
        =(abc+aw^2+bv^2+cu^2)+a(bc+w^2)+b(ac+v^2)+c(ab+u^2)=0\; ,
    \end{align*}
    that is, the coordinate vector $p$ of the principal-minors of $A$ is a singular point with respect to the quadratic form
    \begin{align*}
        Q_{\mathrm{spin}}(p)
        :=p_\emptyset p_{123}+p_1p_{23}+p_2p_{13}+p_3p_{12}\; .
    \end{align*}
    Clearly, $Q_{\mathrm{spin}}$ is hyperbolic and has polar form
    \begin{align*}
        B_{\mathrm{spin}}(p,p')
        =p_\emptyset p'_{123}+p_{123}p'_\emptyset
        +p_1p'_{23}+p_{23}p'_1
        +p_2p'_{13}+p_{13}p'_2
        +p_3p'_{12}+p_{12}p'_3\; .
    \end{align*}
    Substituting the coordinate vectors of principal minors of two symmetric matrices $A$ and $B$, one also verifies
    \begin{align*}
        B_{\mathrm{spin}}\bigl(\sigma(L_A),\sigma(L_B)\bigr)
        =\det(A+B)\; .
    \end{align*}
    Indeed, let $A=\begin{pmatrix}a&u&v\\ u&b&w\\ v&w&c\end{pmatrix}$ and $B=\begin{pmatrix}a'&u'&v'\\ u'&b'&w'\\ v'&w'&c'\end{pmatrix}$. Over $\mathbb F_2$, $\det(A)=abc+aw^2+bv^2+cu^2$, and similarly for $B$. Substituting the principal-minor coordinates of $A$ and $B$ into the polar form $B_{\mathrm{spin}}$ gives
    \begin{align*}
        &\hspace{-.5cm}B_{\mathrm{spin}}\bigl(\sigma(L_A),\sigma(L_B)\bigr)\\
        &=\det(A)+\det(B)
        +a(b'c'+w'^2)+a'(bc+w^2)
        +b(a'c'+v'^2)+b'(ac+v^2)
        +c(a'b'+u'^2)+c'(ab+u^2)\\
        &=(a+a')(b+b')(c+c')
        +(a+a')(w^2+w'^2)
        +(b+b')(v^2+v'^2)
        +(c+c')(u^2+u'^2)\\
        &=(a+a')(b+b')(c+c')
        +(a+a')(w+w')^2
        +(b+b')(v+v')^2
        +(c+c')(u+u')^2\\
        &=\det(A+B)\; ,
    \end{align*}
    where we used $(r+r')^2=r^2+r'^2$ over $\mathbb F_2$. On the other hand, $L_A\cap L_B=\{(x,Ax)\mid(A+B)x=0\}$, hence,
    \begin{align}\label{eq: affine overlaps}
        L_A\cap L_B\neq\{0\}
        \quad\Longleftrightarrow\quad
        \det(A+B)=0
        \quad\Longleftrightarrow\quad
        \sigma(L_A)\perp\sigma(L_B)\; .
    \end{align}

    \textbf{(iii).} By \cite{VanGeemenMarrani2019}, 
    the principal-minor
    map in part~\textup{(ii)} extends to an
    $\Sp(V)$-equivariant bijection
    \begin{align*}
        \sigma:\Lag(V)\overset{\sim}{\longrightarrow}Q^+(7,2)\; ,
    \end{align*}
    where equivariance is with respect to the eight-dimensional spin representation
    \begin{align}\label{eq: spin equivariance}
        \rho:\Sp(V)\longrightarrow O(Q_{\mathrm{spin}})\; ,\qquad
        \sigma(gL)
        =\rho(g)\sigma(L)\; .
    \end{align}
    It remains to prove the correspondence in Eq.~(\ref{eq: affine overlaps}) globally. Let $L,M\in\Lag(V)$ and put $r=\dim(L\cap M)$. Choose a symplectic basis such that $L=\langle e_1,\ldots,e_3\rangle$ and $M=\langle e_1,\ldots,e_r,f_{r+1},\ldots,f_3\rangle$. Then $N=\langle f_1,\ldots,f_r,e_{r+1}+f_{r+1},\ldots,e_3+f_3\rangle$ is a Lagrangian complement of both $L$ and $M$. Choose $g\in\Sp(V)$ such that $gN=Z$. The Lagrangians $gL$ and $gM$ are then both $Z$-transverse and hence belong to a common symmetric-matrix chart.
    
    Since $g\in\Sp(V)$ preserves intersections, while $\rho(g)$ preserves $B_{\mathrm{spin}}$, equivariance in Eq.~\eqref{eq: spin equivariance} and Eq.~\eqref{eq: transverse orthogonality} give
    \begin{align*}
        L\cap M\neq\{0\}
        &\Longleftrightarrow
        gL\cap gM\neq\{0\}\\
        &\Longleftrightarrow
        B_{\mathrm{spin}}
           \bigl(\sigma(gL),\sigma(gM)\bigr)=0\\
        &\Longleftrightarrow
        B_{\mathrm{spin}}
           \bigl(\sigma(L),\sigma(M)\bigr)=0
        \Longleftrightarrow
        \sigma(L)\perp\sigma(M)\; .\qedhere
    \end{align*}
\end{proof}

Our strategy below will now be as follows: using the spin embedding in Eq.~(\ref{eq: spin embedding}) (see Lm.~\ref{lm: spin embedding}, and Ref.~\cite{LevayPlanatSangia2013,VanGeemenMarrani2019} for more details), in Lm.~\ref{lm: spin trichotomy}, we first characterise the condition $L_i\cap L_j\neq\emptyset$ for the elements $|L_i,\xi_i\rangle$ of a set of pairwise orthogonal $3$-qubit stabiliser states (see Eq.~(\ref{eq: stabiliser overlap})) via the corresponding families of points in the embedded space; second, we show that these sets can be extended to a complete stabiliser basis by adjoining characters in Thm.~\ref{thm: 3-qubit extendibility}.

To begin with, recall that a \emph{generator} of $Q^+(7,2)$ is a maximal totally singular projective subspace, that is, it is the projectivisation of a Lagrangian subspace of the ambient symplectic space $\mathbb F_2^8$. Accordingly, a maximal totally singular projective subspace is isomorphic to $\mathrm{PG}(3,2)$ and contains $15$ points.

\begin{lemma}\label{lm: spin trichotomy}
    Every pairwise intersecting set of Lagrangians in $V=\mathbb F^6_2$ is contained in one of three $15$-element families:
    \begin{itemize}
        \item[(I)] $\cP_D=\{L\mid D\subseteq L\}$, \quad for $0\neq D\leq V$ a $1$-dimensional subspace;
        \item[(II)] $\cN_M=\{M\}\cup\{L\mid\dim(L\cap M)=2\}$, \quad for $M\in\Lag(V)$;
        \item[(III)] $\cR_\fq^\pm$, \quad where $\fq$ is a quadratic refinement of $\omega$ of Witt index $3$, and $\cR^\pm_\fq$ denote the two classes of maximal $\fq$-singular Lagrangians (that is, the two classes of generators) of the hyperbolic quadric $Q^+_\fq(5,2)$ (see Tab.~\ref{tab: Klein correspondence}).
    \end{itemize}
\end{lemma}

\begin{proof}
    Under the spin embedding in Eq.~(\ref{eq: spin embedding}), specifically, Eq.~(\ref{eq: spin orthogonality}), the images of pairwise intersecting Lagrangians correspond to pairwise perpendicular singular points of $Q^+(7,2)$. Their span is totally singular and is therefore contained in a generator. It remains to identify the inverse images of the generators of $Q^+(7,2)$.\\
    
    \textbf{(I).} Clearly, any family of elements in $\cP_D$ have pairwise nontrivial intersection. In fact, we have $|\cP_D|=15$ since there is a bijection between Lagrangians in $\cP_D$ and in $\Lag(D^\perp/D)$, where $D^\perp/D$ is a symplectic space of dimension $4$, hence, $|\Lag(D^\perp/D)|=\prod_{i=1}^2(2^i+1)=15$. Note that $\bigcap_{L\in\cP_D}L=D\neq\emptyset$. Moreover, note that since one-dimensional subspaces $D\subset V$ correspond with non-zero elements $0\neq v\in V$, there are $63$ families of the form $\cP_D$.\\

    \textbf{(II).} Fix $M\in\Lag(V)$ and note that $M$ has $7$ two-dimensional subspaces $K$. For every such subspace, there exist $|\Lag(K^\perp/K)|=3$ Lagrangian subspaces containing $K$, including $M$ itself. It follows that $|\cN_M|=1+2\cdot 7=15$. To see explicitly that any two of its members intersect nontrivially, let $K_1=L_1\cap M$ and $K_2=L_2\cap M$ for distinct $L_1,L_2\in\cN_M\backslash\{M\}$. Since $K_1,K_2\subset M$ and $\dim(M)=3$, we must have $\dim(K_1\cap K_2)=1$, which implies $\{0\}\neq K_1\cap K_2\subseteq L_1\cap L_2$. Note that the intersection pattern of the elements in $\cN_M$ differs from that in $\cP_D$ in that $\dim(L\cap M)=2$ for all $L\in\cN_M$ and $\bigcap_{L\in\cN_M}L=\{0\}$. Moreover, we clearly have $|\{\cN_M\}_{M\in\Lag(V)}|=135=|\Lag(V)|$.\\

    \textbf{(III).} Let $\fq$ be a quadratic refinement of the symplectic form $\omega$ of $V=\mathbb F^6_2$ with Witt index $3$. The $\fq$-singular Lagrangians correspond projectively with the maximal singular planes of $Q^+_\fq(5,2):=\{[v]\in\PG(5,2)\mid \fq(v)=0\}$. Now, a hyperbolic quadric $Q^+_\fq(5,2)$ has $30$ maximal singular planes, which split into two types $\cR^\pm_\fq$ with $|\cR^\pm_\fq|=15$; indeed, they correspond with the $\alpha$ and $\beta$ families in the Klein correspondence (see Tab.~\ref{tab: Klein correspondence}), from which it also follows that $\dim(L\cap L')=1$ for all distinct $L,L'\in\cR^\pm_\fq$ (as well as $\bigcap_{L\in\cR^\pm_\fq}L=\{0\}$). For fixed $\fq$, the two generator classes $\cR^+_\fq$ and $\cR^-_\fq$ are distinct. Moreover, either generator class covers the entire singular point set of $Q^+_\fq(5,2)$. Hence, if $\cR^s_\fq=\cR^{s'}_{\fq'}$, then $\fq$ and $\fq'$ have the same singular point set and therefore $\fq=\fq'$, and for this fixed refinement, equality then forces $s=s'$. We can thus count the number of such families, by recalling that every quadratic refinement $\fq$ is of the form $\fq(v)=\fq_a(v)=\fq_0(v)+\omega(a,v)$ for $a\in V$, from which it follows that there are $|V|=64$ quadratic refinements in total. The cardinality of those with Witt index $3$ corresponds with the vectors for which $\fq_0(a)=0$, and is given by $36$. Adding the above two types, thus yields a total of $2\cdot 36=72$ families of type (III).

    Finally, we need to show that the above families exhaust all families of pairwise intersecting Lagrangians. This follows since the families in the three types are mutually distinct, as they are distinguished by their intersection pattern (a nonzero common intersection; a unique member meeting every other member in dimension two; and the property that every two distinct members meet in dimension one), and since the number of generators of $Q^+(7,2)$ is $270$, matching the sum of the families of types (I), (II), (III) above, $63+135+72=270$.
\end{proof}

By Eq.~(\ref{eq: stabiliser overlap}) and Lm.~\ref{lm: spin trichotomy}, the Lagrangians corresponding to a set of pairwise orthogonal stabiliser states fall within three families. To prove that they always extend to a stabiliser basis, we need to add the character data. For the families of type (I) and (II) this is straightforward; type (III) requires more work, yet reveals an intricate connection with the root lattice of $E_8$. We summarise the correspondence in the following lemma (for details, see Ref.~\cite{LevayPlanatSangia2013,VanGeemenMarrani2019}).

\begin{lemma}\label{lm: RM-E8 correspondence}
    Let $X=\mathbb F_2^3$ and let $\{e_x\}_{x\in X}$ be the standard basis of $\R^X\cong\R^8$. Define the set,
    \begin{align*}
        H_8=\mathrm{RM}(1,3)
        =\left\{\bigl(a\cdot x+b\bigr)_{x\in X}\mid a\in\mathbb F_2^3,\ b\in\mathbb F_2
        \right\}\; .
    \end{align*}    
    Then the following hold.
    \begin{enumerate}
        \item[(i)] $H_8$ is the binary $[8,4,4]$ extended Hamming code. Its nonzero codewords consist of the all-one word and fourteen weight-four words. The supports of the weight-four words are precisely the affine
        hyperplanes of $X$.
        
        \item[(ii)] The vectors of $H_8$ define the root lattice
        \begin{align*}
            \Lambda_{E_8}
            =\frac1{\sqrt2}\{z\in\mathbb Z^X\mid z\bmod2\in H_8\}
        \end{align*}
        of the exceptional Lie algebra $E_8$, which correspond with $16$ basis vectors $\{\pm\sqrt2\,e_x\}_{x\in X}$ and $14\cdot 16=224$ roots in
        \begin{align*}
            \left\{\frac1{\sqrt2}\sum_{x\in H}\epsilon_xe_x\mid H\subseteq X\text{ an affine hyperplane},\ \epsilon_x\in\{\pm1\}\right\}\; .
        \end{align*}
        
        \item[(iii)] Let $\overline\Lambda=\Lambda_{E_8}/2\Lambda_{E_8}$, and define the quadratic form with corresponding polar form,
        \begin{align*}
            Q(\overline v)=\frac{(v,v)}2\pmod2\; ,\qquad
            B(\overline v,\overline w)=(v,w)\pmod2\; .
        \end{align*}
        Then $(\overline\Lambda,B)$ is an $8$-dimensional symplectic vector space, $Q$ a quadratic refinement of Witt index $4$, and
        \begin{align}\label{eq: E8 correspondence}
            \{\text{$E_8$-root rays}\}
            \overset{\sim}{\longrightarrow}
            \{u\in\overline\Lambda\mid Q(u)=1\}\; ,
        \end{align}
        is an orthogonality-preserving bijection.
    \end{enumerate}
\end{lemma}

\begin{proof}
    \textbf{(i)} Let $X=\mathbb F_2^3$. By definition, $H_8=\mathrm{RM}(1,3)=\left\{\bigl(a\cdot x+b\bigr)_{x\in X}\mid a\in\mathbb F_2^3,\ b\in\mathbb F_2\right\}$ (see, for example, \cite[\S1]{MiezakiMunemasa2024}). There are $2^4=16$ affine linear functions, the zero function has Hamming weight zero, the constant-one function has Hamming weight eight, and every nonconstant affine function has a kernel of size four and hence Hamming weight four. Thus the fourteen weight-four supports are precisely the affine hyperplanes of $X$. In particular, $H_8$ is doubly even, that is, $\mathrm{wt}(c)=0\pmod4$ for all $c\in H_8$, where $\mathrm{wt}$ denotes the Hamming weight of the codeword $c\in H_8$.
    
    Since every codeword has Hamming weight divisible by four, for two codewords $c,d$ one has
    \begin{align*}
        c\cdot d
        =\sum_{i=1}^8c_id_i
        =|\mathrm{supp}(c)\cap\mathrm{supp}(d)|
        =\frac{\mathrm{wt}(c)+\mathrm{wt}(d)-\mathrm{wt}(c+d)}2\pmod2\; .
    \end{align*}
    It follows that $H_8$ is self-orthogonal (that is, $H_8\subseteq H^\perp_8)$, and since $\dim H_8=4$ and its length is eight, it is also self-dual, that is, $H_8=H^\perp_8$. Hence, $H_8$ is the doubly-even self-dual $[8,4,4]$ extended Hamming code.
    
    \textbf{(ii)} Define the lattice corresponding to the binary code $H_8$ by (see Ref.~\cite[chapter 7]{ConwaySloane2013})
    \begin{align*}
        \Lambda
        &=\frac{1}{\sqrt{2}}\{z\in\mathbb Z^8\mid z\bmod2\in H_8\}\; .
    \end{align*}
    Clearly, $\Lambda$ has full rank since $2\zz^8\subseteq\{z\in\mathbb Z^8\mid z\bmod2\in H_8\}$, hence, $\sqrt{2}\zz^8\subseteq\Lambda$. Moreover, $\Lambda$ is integral since for any $z,w\in\zz^8$ with $H_8\ni c=z\pmod2$ and $H_8\ni d=w\pmod2$, we have $\left(\frac z{\sqrt2},\frac w{\sqrt2}\right)=\frac{z\cdot w}{2}\in\mathbb Z$ since $c\cdot d=0$, by self-orthogonality of $H_8$. Next, note that $z\cdot z=\sum_{i=1}^8z^2_i\equiv\mathrm{wt}(c)\pmod4$, since $H_8$ is doubly even, from which it follows that $\Lambda$ is even, that is, $\left(\frac z{\sqrt2},\frac z{\sqrt2}\right)=\frac{z\cdot z}{2}\in 2\zz$ for every $z\in\Lambda$. The inverse image of $H_8$ in $\mathbb Z^8$ has index $[\zz^8:\sqrt{2}\Lambda]=2^{8-\dim H_8}=16$. Scaling by $\frac{1}{\sqrt{2}}$ in eight dimensions multiplies the covolume (its square is the determinant of the Gram matrix of a lattice basis) by $2^{-4}=\frac{1}{16}$. Hence, $\Lambda$ has covolume one and is thus unimodular (as a full-rank lattice), that is, $\Lambda^*=\Lambda$. It is therefore an even positive-definite unimodular lattice of rank eight, and consequently $\Lambda\cong\Lambda_{E_8}$ by the uniqueness theorem for such lattices (see Ref.~\cite{Elkies2004}).
    
    A root of $\Lambda$ has norm two, so its integer representative $z$ satisfies $z\cdot z=4$. There are only two possibilities: either $z=\pm2e_x$, giving the sixteen roots $\{\pm\sqrt{2}e_x\}_{x\in X}$, or $z$ has four coordinates equal to $\pm1$ and the others zero. Its odd-coordinate support therefore is a weight-four codeword, hence an affine hyperplane $H\subseteq X$, which gives $\{\frac{1}{\sqrt{2}}\sum_{x\in H}\epsilon_xe_x\mid H\subset X\text{ an affine hyperplane},\ \epsilon_x\in{\pm1}\}$. There are fourteen choices of $H$ and sixteen sign patterns, giving $14\cdot16=224$ such roots, hence, $16+224=240$ roots in total (and thus $120$ rays). After rescaling to unit vectors these $240$ roots correspond exactly with the sets $\{\pm e_x\}_{x\in X}$ and $\left\{\frac{1}{2}\sum_{x\in H}\epsilon_xe_x\mid H\subseteq X\text{ an affine hyperplane},\ \epsilon_x\in{\pm1}\right\}$.
    
    \textbf{(iii)} We first prove the claimed properties of the quadratic form $Q$ and its polarisation.
    
    Replacing $v$ by $v+2a$ changes $(v,v)/2$ by $2(v,a)+2(a,a)=0\pmod 2$, hence, $Q$ is well defined. Clearly, $B$ is its polar form. It is alternating since $\Lambda$ is even and thus $B(\overline{v},\overline{v})=(v,v)=0\pmod2$. If $B(\overline v,\overline w)=0$ for all $\overline w$, then $(v,w)\in 2\zz$, equivalently $(\frac{v}{2},w)\in\zz$, hence, $\frac{v}{2}\in\Lambda^*$. By unimodularity, this is equivalent to $v\in 2\Lambda$, hence, $\overline v=0$, and $B$ is nondegenerate. Finally, to determine the type of the quadratic refinement $Q$ of $B$, let $X=\mathbb F_2^3$, let $\epsilon_1,\epsilon_2,\epsilon_3$ denote its standard basis, and let $\{e_x\}_{x\in X}$ denote the standard orthonormal basis of $\mathbb R^X$, such that $e_x(y)=\delta_{x,y}$. Define
    \begin{align*}
        r_0:=\frac1{\sqrt2}\sum_{x\in X}e_x,
        \qquad
        r_i:=\sqrt2(e_0+e_{\epsilon_i}),
        \quad i=1,2,3.
    \end{align*}
    Clearly, these vectors belong to $\Lambda_{E_8}$. Indeed, the integer representative of $r_0$ is the all-one vector, whose reduction modulo two is the constant-one codeword in $H_8=\operatorname{RM}(1,3)$. The integer representative of $r_i$ is $2e_0+2e_{\epsilon_i}$, whose reduction modulo two is the zero codeword. Moreover, one checks that $(r_i,r_i)=4$, $(r_0,r_i)=2$ and $(r_i,r_j)=2$ for all $i\neq j$. Consequently, their reductions $\bar r_i\in\overline\Lambda:=\Lambda_{E_8}/2\Lambda_{E_8}$ satisfy $Q(\bar r_i)=0$ and are pairwise orthogonal, that is, $B(\bar r_i,\bar r_j)=0$. To prove linear independence, suppose $\sum_{i=0}^3 a_i\bar r_i\in 2\Lambda_{E_8}$. Reduction of integer representatives modulo two immediately gives $a_0=0$. The remaining relation implies $\sum_{i=1}^3 a_i(e_0+e_{\epsilon_i})\in H_8$. Every sum of one or two of these vectors has weight two, while their total sum has support $\{0,\epsilon_1,\epsilon_2,\epsilon_3\}$, which is not an affine hyperplane, as it contains $0$ and the three independent vectors $\epsilon_1,\epsilon_2,\epsilon_3$. Thus the sum belongs to $H_8$ only when $a_1=a_2=a_3=0$.

    Consequently, $\langle\bar r_0,\bar r_1,\bar r_2,\bar r_3\rangle$ is a $4$-dimensional totally singular subspace of the $8$-dimensional symplectic space $\Lambda$, contained in $Q$, hence, $Q$ has Witt index four, from which it follows that $|\{u\in\overline{\Lambda}\mid Q(u)=1\}|=120$ \cite{Taylor1992}.

    Having established these facts, we now prove that the map in Eq.~(\ref{eq: E8 correspondence}) is an orthogonality-preserving bijection.
    
    If $r$ is a root, then $Q(\overline r)=\frac{(r,r)}2=1$, and $\overline r=\overline{-r}$, so reduction depends only on the root ray. Suppose roots $r,s$ have the same reduction. Then $r-s=2a$ for some $a\in\Lambda$, and hence $(r-s,r-s)=4(a,a)$ is divisible by eight. By contrast, for distinct nonopposite $E_8$ roots, $(r,s)\in\{-1,0,1\}$, and thus $(r-s,r-s)=4-2(r,s)\in\{2,4,6\}$. Consequently, the reduction is injective on root rays, and since $\{u\in\overline{\Lambda}\mid Q(u)=1\}$ has cardinality $120$, Eq.~(\ref{eq: E8 correspondence}) is a bijection.
    
    Finally, we show that Eq.~(\ref{eq: E8 correspondence}) preserves orthogonality. Clearly, if $(r,s)=0$ for $r,s\in\Lambda$ , then also $B(\overline{r},\overline{s})=(r,s)=0\pmod2$. Conversely, let $\overline{r},\overline{s}\in\overline{\Lambda}$ with $B(\overline{r},\overline{s})=0$. Then $(r,s)=0\pmod 2$, and since the inner product between two distinct, nonopposite $E_8$ roots takes values $-1$, $0$ or $1$ only, the only even possibility is $(r,s)=0$.
\end{proof}

With the above correspondence at hand, we are now in the position to prove the remaining extendibility result.

\begin{theorem}\label{thm: 3-qubit extendibility}
    Every set of pairwise orthogonal three-qubit stabiliser states extends to a stabiliser basis.
\end{theorem}

\begin{proof}
    Let $\cF=\{|L_i,\xi_i\}_{i=1}^k$ be such a set and let $\cL=\{L_i\}_{i=1}^k$ be the set of its distinct Lagrangian subspaces. By Eq.~(\ref{eq: stabiliser overlap}), we must have $L\cap L'\neq\{0\}$ for all $L,L'\in\cL$, and, by Lm.~\ref{lm: spin trichotomy}, $\cL$ is contained in a family of type (I-III).\\
    
    \textbf{(I).} In this case, $D\subset\bigcap_i L_i$. After a Clifford transformation taking $D$ to $Z_1$, every state has the form $\ket{z_t}_1\otimes\ket{\phi}_{23}$, for $t\in\mathbb F_2$. For each value of $t$, the residual states are pairwise orthogonal two-qubit stabiliser states. By Thm.~\ref{thm: 2 qudit extendibility}, each branch extends to a two-qubit stabiliser basis, and the union of the two branches thus constitutes a stabiliser basis.\\
    
    \textbf{(II).} Let $\cB_M=\{\ket{x}_M\mid x\in\Xi(M)\}$ be the joint eigenbasis of $M$. A stabiliser state $|M,\xi\rangle$ then has support on precisely one element of $\cB_M$. If $M\neq L\in\cN_M$, then $\dim(L\cap M)=2$, and Lm.~\ref{lm: affine stabiliser support} implies that every stabiliser state $|L,\xi\rangle$ has support on a two-point affine line $\{x,y\}\subset\Xi(M)$. It therefore has the form
    \begin{align*}
        \frac{\ket{x}_M+\zeta\ket{y}_M}{\sqrt2}\; ,\qquad
        \zeta\in\{\pm1,\pm i\}\; .
    \end{align*}
    It follows that two distinct affine lines cannot share an element of $\cB_M$; similarly, a stabiliser state $|M,\xi\rangle$ cannot coexist with an edge containing that element. (In both cases, the orthogonality condition is violated, since the respective inner products contain a nonzero summand.) Consequently, the affine stabiliser supports in $\cF$ decompose into mutually disjoint singletons and pairs of elements in $\cB_M$. Now, to complete $\cF$ to a stabiliser basis, for any state with two-point support in $\cF$ add the unique orthogonal partner $\frac{\ket{x}_M-\zeta\ket{y}_M}{\sqrt2}$, and adjoin all $M$-basis states not contained in such a support or already in $\cF$. After this procedure, every affine line thus contributes two states, supported on two elements of $\cB_M$, and every remaining element of $\cB_M$ contributes one additional state. Disjointness implies that this yields a total of eight pairwise orthogonal stabiliser states, hence, a stabiliser basis.\\
    
    \textbf{(III).} Without loss of generality, let $\cL\subset\cR_\fq^s$ for $s\in\{\pm\}$. After a Clifford transformation, we may take $\fq(v)=x\cdot z$ for $v=(x,z)$ (the quadratic refinement in the Weyl representation \cite{Gross2006}). The $\fq$-singular Pauli operators are then represented by real symmetric matrices, hence, all joint eigenvectors of the Lagrangians in $\cR_\fq^s$ may be chosen real.
    
    Next, fix any $M\in\cR_\fq^s$ and let $\{e_x\}_{x\in\mathbb F_2^3}$ be its real eigenbasis. Moreover, the $120$ stabiliser states associated with the $15$ elements of $\cR_\fq^s$ can further be identified with the set
    \begin{align}\label{eq: type-III stabiliser states}
        \{e_x\}_{x\in\mathbb F_2^3}\cup
        \left\{\R\left(\frac{1}{2}\sum_{x\in H}\epsilon_xe_x\right)\bigg\mid H\subseteq\mathbb F_2^3\text{ an affine hyperplane},\ \epsilon_x\in\{\pm1\}\right\}\; ,
    \end{align}
    where sign patterns differing by an overall sign define the same state. Indeed, $M$ gives the eight coordinate rays, and for every $M\neq L\in\cR^s_\fq$, we have $\dim(L\cap M)=1$ by Lm.~\ref{lm: spin trichotomy}, hence, the stabiliser states associated with $L$ have four-point affine support (see Lm.~\ref{lm: affine stabiliser support}) and real coefficients of magnitude $\frac{1}{2}$ by Eq.~(\ref{eq: stabiliser overlap}). Consequently, the fourteen affine hyperplanes and their eight projective sign patterns in Eq.~(\ref{eq: type-III stabiliser states}) represent the remaining $14\cdot8=112$ rays.
    
    With Lm.~\ref{lm: RM-E8 correspondence}, (i), we may identify the incidence vectors of the affine hyperplanes of $\mathbb F_2^3$ in Eq.~(\ref{eq: type-III stabiliser states}) with the fourteen weight-four codewords of the extended binary Hamming code $H_8=\mathrm{RM}(1,3)$, and after multiplying the representatives in $H_8$ by $\frac{1}{\sqrt{2}}$, they correspond, up to an overall sign, with the roots of the lattice
    \begin{align*}
        \Lambda_{E_8}
        =\frac{1}{\sqrt2}\{z\in\mathbb Z^8\mid z\bmod2\in H_8\}\; .
    \end{align*}
    By Lm.~\ref{lm: RM-E8 correspondence} (ii), $\Lambda_{E_8}$ constitutes a representation of the root lattice of the exceptional Lie algebra $E_8$. In particular, the elements in Eq.~(\ref{eq: type-III stabiliser states}) correspond with the set of $120$ rays of the $E_8$ root system.
    
    It thus remains to show that every orthogonal family of $E_8$-root rays extends to an orthogonal basis of eight root rays. We will show this analytically using the above correspondence (for an earlier computational proof, see also Ref.~\cite{RuugeVanOystaeyen2005,WaegellAravind2015}). To see this, we use that, by Lm.~\ref{lm: RM-E8 correspondence} (iii), the $120$ root rays are a bijection with the nonsingular elements of the quadric $Q(\overline v)=\frac{(v,v)}2\pmod2$ of hyperbolic type in an $8$-dimensional symplectic vector space, preserving orthogonality with respect to its polar symplectic form $B(\overline v,\overline w)=(v,w)\pmod2$. We may thus represent the stabiliser states in $\cF$ by vectors $u_1,\ldots,u_k\in\{u\in\overline{\Lambda}\mid Q(u)=1\}$. Pairwise orthogonality implies $B(u_i,u_j)=0$ for all $i\neq j$. Now, extend their span to a maximal $B$-isotropic subspace $T\leq\overline\Lambda$, so that $\dim T=4$. Since $B|_T=0$, the restriction $Q|_T$ is linear, and it is nonzero because $Q(u_i)=1$ for all $1\leq i\leq k$. Therefore, $\{u\in T\mid Q(u)=1\}$ is an affine hyperplane of $T$ containing eight vectors. Finally, since orthogonality is preserved under Eq.~(\ref{eq: E8 correspondence}), the corresponding root rays are pairwise orthogonal. Since $\{u\in T\mid Q(u)=1\}$ contains the image of the original family under the correspondence in Eq.~(\ref{eq: E8 correspondence}) and forms an orthogonal basis of $\R^8$, it defines a stabiliser basis under the reverse correspondence.
\end{proof}

The proof of Thm.~\ref{thm: 3-qubit extendibility} reveals why the existence results for qubit and odd-prime-dimensional qudit USBs differ: neither the spin embedding $\sigma$ in Lm.~\ref{lm: spin embedding} nor the identification of stabiliser states, corresponding to the Lagrangians in $\cR^s_\fq$, with roots of the $E_8$ lattice, corresponding with nonsingular points of a quadratic refinement of Witt index $4$ of an $8$-dimensional symplectic space, carry over to the odd prime case.

\subsection{Lifting property}\label{app: lifting}

Given an unextendible set of pairwise orthogonal stabiliser states for $n$ parties, we construct one for $n+1$ parties.

\begin{lemma}\label{lm: lifting}
    Let $\cU$ be an unextendible set of $n$-qudit stabiliser states of prime local dimension $q$, and let $\{|t\rangle_Z\mid t\in\mathbb F_q\}$ be the computational stabiliser basis of a single qudit. Then the following set is unextendible,
    \begin{align*}
        \cU^\uparrow
        :=\left\{|\psi\rangle\otimes|z_t\rangle\mid|\psi\rangle\in\cU,\ t\in\mathbb F_q
        \right\}\; .
    \end{align*}
\end{lemma}

\begin{proof}
    Clearly, the states in $\cU^\uparrow$ are stabiliser states and are pairwise orthogonal. Moreover, $\mathrm{span}(\cU^\uparrow)=\mathrm{span}(\cU)\otimes\mathbb C^q$. Writing $\cK:=\mathrm{span}(\cU)^\perp$, we have $\mathrm{span}(\cU^\uparrow)^\perp=\cK\otimes\mathbb C^q$. Suppose now that a stabiliser state $|\Phi\rangle$ belonged to this complementary subspace, and measure the Pauli $Z$ operator on the final qudit. At least one postselected branch
    \begin{align*}
        |\phi_t\rangle
        :=\left(\one\otimes\langle z_t|\right)|\Phi\rangle
    \end{align*}
    is nonzero. Pauli measurement and postselection map stabiliser states to stabiliser states, hence, the normalized branch state $|\phi_t\rangle$ is an $n$-qudit stabiliser state. Moreover, since $|\Phi\rangle\in\cK\otimes\mathbb C^q$, the normalised state $|\phi_t\rangle$ lies in $\cK$. Yet, this contradicts unextendibility of $\cU$, hence, the complement of $\cU^\uparrow$ contains no stabiliser state.
\end{proof}

Combining the explicit constructions of USBs in Thm.~\ref{thm: 4 qubit USB} and Thm.~\ref{thm: 3 qudit USB} with the general extendibility results in Thm.~\ref{thm: 2 qudit extendibility} and Thm.~\ref{thm: 3-qubit extendibility}, and the lifting argument in Lm.~\ref{lm: lifting} thus proves the general existence result for USBs in Thm.~\ref{thm: unextendibility}.

\section{Magic witness, uncompletability, operational indistinguishability and bound magic}\label{app: applications}

We discuss various applications of USBs for the resource theory of magic. In this regard, we take stabiliser operations (SO) to mean finite adaptive stabiliser protocols generated by stabiliser-state preparation, Clifford unitaries, projective Pauli measurements, discarding, and classical randomness/conditioning \cite{VeitchEtAl2014,HowardCampbell2017}.\\

\textbf{Magic witnesses.} Similarly to testing separability using entanglement witnesses, certifying whether a state is magic can be done using magic witnesses, that is, via Hermitian operators $W$ with the property that for every stabiliser state $\sigma\in\Stab_{n,q}$, $\tr[W\sigma]\geq0$, whereas $\tr[W\rho]<0$ for some density operator $\rho\notin\Stab_{n,q}$.

\begin{lemma}\label{lm: magic witness}
    Let $\cK\subseteq\cH$ be a subspace containing no stabiliser state, let $\Pi_\cK$ be its orthogonal projector, and define
    \begin{align*}
        \alpha_\cK
        :=\max_{\ket{s}\in\Stab^\mathrm{pure}_{n,q}}
        \bra{s}\Pi_\cK\ket{s}\; .
    \end{align*}
    Then $\alpha_\cK<1$ and $W_\cK:=\alpha_\cK\one-\Pi_\cK$ is a magic witness, in particular, every state supported on $\cK$ is magic.
\end{lemma}

\begin{proof}
    Since the set of pure stabiliser states is finite, the maximum defining $\alpha_\cK$ is attained. Moreover, $\bra{s}\Pi_\cK\ket{s}=1$ holds if and only if $\ket{s}\in\cK$. By assumption, $\cK$ contains no stabiliser state, and hence $\alpha_\cK<1$.
    
    For every pure stabiliser state $\ket{s}\in\Stab^\mathrm{pure}_{n,q}$, $\bra{s}W_\cK\ket{s}=\alpha_\cK-\bra{s}\Pi_\cK\ket{s}\geq0$, and by convexity, the same inequality holds for every $\sigma\in\Stab_{n,q}$. Finally, if $\rho$ is supported on $\cK$, then
    $\Pi_\cK\rho=\rho$, and therefore
    \begin{equation*}
        \tr[W_\cK\rho]
        =\alpha_\cK\tr[\rho]
        -\tr[\Pi_\cK\rho]
        =\alpha_\cK-1<0\; .\qedhere
    \end{equation*}
\end{proof}

In particular, every USB naturally defines a magic witness.

\begin{corollary}\label{cor: USB magic witness}
    Let $\cU=\{\ket{\psi_i}\}_{i=1}^k$ be a USB and set $P_{\cU}:=\sum_{i=1}^k\dyad{\psi_i}$. Then $W_{\cU}:=P_{\cU}-\delta_{\cU}\one$ is a magic witness with
    \begin{align*}
        \delta_{\cU}
        :=\min_{\ket{s}\in\Stab^\mathrm{pure}_{n,q}}
        \bra{s}P_{\cU}\ket{s}
        =\min_{\ket{s}\in\Stab^\mathrm{pure}_{n,q}}
        \sum_{i=1}^k|\langle\psi_i|s\rangle|^2
        >0\; .
    \end{align*}
\end{corollary}

\begin{proof}
    Indeed, $\tr[W_{\cU}\sigma]\geq0$ for all $\sigma\in\Stab_{n,q}$, and $\tr[W_{\cU}\rho]=-\delta_{\cU}<0$ for every state $\rho$ supported in $\mathrm{ran}(Q_{\cU})$.
\end{proof}

\textbf{Uncompletability.} Define a \emph{Clifford isometry} $J:\cH\ra\cH\otimes\cH_A$ to be an embedding of the form
\begin{align*}
    J|\psi\rangle
    =C\left(|\psi\rangle\otimes|\alpha\rangle_A\right)\; ,
\end{align*}
where $|\alpha\rangle_A$ is a stabiliser ancilla and $C$ a Clifford unitary.

\begin{lemma}\label{lm: stabiliser uncompletability}
    Let $\cU$ be an unextendible stabiliser basis, and $J$ a Clifford isometry. Then $J(\cU)$ cannot be completed to an orthonormal basis of pure qudit stabiliser states of the enlarged system.
\end{lemma}

\begin{proof}
    Let $J:\cH\ra\cH\otimes\cH_A$, $J|\psi\rangle=C\left(|\psi\rangle\otimes|\alpha\rangle_A\right)$ with $|\alpha\rangle_A$ a stabiliser ancilla and $C$ a Clifford unitary. Suppose to the contrary that there exist pure stabiliser states $\{|\phi_a\rangle\}_{a=1}^{D_AD-k}$ such that $\{J|\psi_i\rangle\}_{i=1}^k\cup{|\phi_a\rangle}_{a=1}^{D_AD-k}$ is a complete orthonormal stabiliser basis. With Eq.~(\ref{eq: unextendible subspace}), we then have $\one-JP_\cU J^\dagger=\sum_a\dyad{\phi_a}$, as well as
    \begin{align*}
        Q_\cU
        =J^\dagger(\one-JP_\cU J^\dagger)J
        =\sum_a\dyad{\widetilde\phi_a}\; ,
    \end{align*}
    where $|\widetilde\phi_a\rangle:=J^\dagger|\phi_a\rangle=(\one\otimes{}_A\langle\alpha|)C^\dagger|\phi_a\rangle$. Since pure stabiliser states are closed under Clifford transformations and stabiliser postselection, every $|\widetilde\phi_a\rangle$ is either zero or proportional to a pure stabiliser state. Yet from
    \begin{align*}
        \langle\psi_i|\widetilde\phi_a\rangle
        =\langle\psi_i|J^\dagger|\phi_a\rangle
        =\langle J\psi_i|\phi_a\rangle
        =0\qquad\forall 1\leq i\leq k,\ 1\leq a\leq D_AD-k\; .
    \end{align*}
    and $|\cU|<D$, we have $Q_\cU\neq0$, and the above decomposition thus contains at least one nonzero vector $|\widetilde\phi_a\rangle$. Its normalisation is a pure stabiliser state orthogonal to every member of $\cU$, contradicting that $\cU$ is a USB.
\end{proof}

In other words, while unextendibility is not preserved under a proper stabiliser embedding, e.g. after the embedding $|\psi_i\rangle\longmapsto|\psi_i\rangle\otimes|0\rangle_A$, every state $|s\rangle\otimes|1\rangle_A$ is an orthogonal stabiliser extension, Lm.~\ref{lm: stabiliser uncompletability} shows that the embedded family is \emph{uncompletable}: it cannot be enlarged to a complete stabiliser basis of the enlarged Hilbert space.\\

\textbf{Distillability.} For a UPB, bound entanglement of $\rho^{\perp}_\cU$ rests on the fact that positivity under partial transposition is (i) stable under tensor powers and (ii) an obstruction to distillation~\cite{Peres1996,Horodeckisz1998}. It is natural to ask whether the resource theory of magic supplies an obstruction that is also stable under tensor products and plays a similar role. For qudits of odd prime local dimension this is the case, and in the strongest possible form: the complement of \emph{every} orthogonal family of stabiliser states is positively represented in phase space.

We recall the discrete Wigner formalism for odd $q$~\cite{Gross2006,Veitch2012,MariEisert2012,VeitchEtAl2014}. For $u\in V=\zz^{2n}_q$, the phase-point operators,
\begin{equation}
    A_u
    :=\frac{1}{D}\sum_{v\in V}\zeta_q^{-\omega(u,v)}W_v
    =W_uA_0W^\dagger_u\; ,
\end{equation}
are Hermitian and satisfy $\tr[A_u]=1$ and $\tr[A_uA_{u'}]=D\delta_{u,u'}$. Here, $A_0=\frac{1}{D}\sum_{v\in V}W_v$ is the parity operator, acting by $A_0\ket{j}=\ket{-j}$ on the computational basis. The discrete Wigner function of an operator $O$ is $W_O(u):=D^{-1}\tr[A_uO]$, so that $\sum_{u\in V}W_O(u)=\tr[O]$. Write $\cW_+:=\{\rho\mid W_\rho\geq 0\}$ for the set of positively represented states. Then the following hold: (i) $\mathrm{Stab}_{n,q}\subseteq\cW_+$, and by the discrete version of Hudson's theorem~\cite{Gross2006}, the pure states in $\cW_+$ are exactly the pure stabiliser states $\ket{A}$, whose Wigner function is the normalised indicator function of its affine Lagrangian $A$, that is, $W_{\Pi_A}=D^{-1}\mathbf{1}_A$ where $\mathbf{1}_A(v)=1(0)$ if $v\in(\notin) A$; (ii) $\cW_+$ is convex, closed under tensor products and preserved by stabiliser operations \cite{Veitch2012,MariEisert2012}; and (iii) the \emph{mana} $\cM(\rho):=\log\sum_{u\in V}\lvert W_\rho(u)\rvert$ is a magic monotone which vanishes precisely on $\cW_+$ and is strictly positive on every pure magic state~\cite{VeitchEtAl2014}. Consequently, no state in $\cW_+$ can be converted by stabiliser operations into a pure magic state, that is, every magic state in $\cW_+$ is bound magic.

\begin{lemma}\label{lm: Wigner positivity}
    Let $q$ be an odd prime and let $\cU=\{\ket{A_i}\}_{i=1}^{k}$, $k<D$, be a set of pairwise orthogonal $n$-qudit stabiliser states, with $Q_\cU$ and $\rho^{\perp}_\cU=Q_\cU/\tr[Q_\cU]$ as in Eq.~\eqref{eq: unextendible subspace}. Then $\rho^{\perp}_\cU$ is represented by the uniform distribution,
    \begin{equation}\label{eq: Wigner positivity}
        W_{\rho^{\perp}_\cU}(u)
        \ =\ \frac{1}{D(D-k)}\Big(1-\mathbf{1}_{\bigcup_i A_i}(u)\Big)
        \ \geq\ 0\; ,
    \end{equation}
    supported on the phase points not covered by the affine Lagrangians $A_i$. In particular, $\rho^{\perp}_\cU\in\cW_+$ and $\cM(\rho^{\perp}_\cU)=0$.
\end{lemma}

\begin{proof}
    By Cor.~\ref{cor: affine stabiliser overlap}, two stabiliser states are orthogonal if and only if their affine Lagrangians are disjoint. Pairwise orthogonality of $\cU$ therefore means that $A_1,\dots,A_k$ are pairwise disjoint, hence, $\sum_i\mathbf{1}_{A_i}=\mathbf{1}_{\bigcup_iA_i}\le 1$ pointwise. Since $O\mapsto W_O$ is linear, $W_\one\equiv D^{-1}$ and $W_{\Pi_{A_i}}=D^{-1}\mathbf{1}_{A_i}$, we obtain $W_{Q_\cU}=D^{-1}(1-\mathbf{1}_{\bigcup_iA_i})\geq 0$. Dividing by $\tr[Q_\cU]=D-k$ gives Eq.~(\ref{eq: Wigner positivity}). Finally, summing over the $q^{2n}=D^2$ phase space points yields $\frac{D^2-kD}{D(D-k)}=1$, hence, $\cM(\rho^{\perp}_\cU)=0$.
\end{proof}

\begin{corollary}\label{cor: bound magic states}
    Let $q$ be an odd prime and let $\cU$ be a USB. Then $\rho^{\perp}_\cU$ is a bound magic state.
\end{corollary}

\begin{proof}
    That $\rho^{\perp}_\cU\notin\mathrm{Stab}_{n,q}$ is shown in Sec.~\ref{sec: UPBs vs USBs}; positivity is Lm.~\ref{lm: Wigner positivity}. As $\cM$ is a magic monotone vanishing on $\rho^{\perp}_\cU$ and strictly positive on every pure magic state, no stabiliser protocol converts copies of $\rho^{\perp}_\cU$ into pure magic states.
\end{proof}

Cor.~\ref{cor: bound magic states} is thus an exact analogue, in the resource theory of magic, of the UPB mechanism for producing bound entanglement~\cite{BennettEtAl1999}, with Wigner positivity taking over the role of positive partial transpose. For the three-qudit USB $\cS_q$ of Thm.~\ref{thm: 3 qudit USB} one has $\tr[Q_{\cS_q}]=q^3-|\cS_q|=2q(q-1)$, hence, $\rho^{\perp}_{\cS_q}$ is represented by the uniform distribution on $2q^{4}(q-1)$ out of the $q^6$ phase space points; for $q=3$ this is a rank-$12$ state supported uniformly on $324$ of the $729$ points.

For qubits, no counterpart of Lm.~\ref{lm: Wigner positivity} is available. This is not an artefact of our construction, but a consequence of the fact that there is no Clifford-covariant quasiprobability representation in which all qubit stabiliser states are non-negatively represented, the obstruction being the state-independent contextuality of the multi-qubit Pauli group, as exhibited by Mermin's square~\cite{HowardEtAl2014,DelfosseEtAl2015,Raussendorf2019}. (Note that the substitutes available in the qubit setting---$\Lambda$-polytopes and CNC-operator representations~\cite{Raussendorf2019,RaussendorfEtAl2023,ZurelOkayRaussendorf2020}---represent \emph{all} states positively and hence single out no free set.)

We provide a concrete counterexample to the analogous claim for bound magic in the qubit case, using the complementary state of the four-qubit USB. Let $\cU_4=\{|u_i\rangle\}_{i=1}^8$ denote the eight states in Thm.~\ref{thm: 4 qubit USB}, in the order listed there, and define $\cP_4=\sum_{i=1}^8\dyad{u_i}$, $\cQ_4=\one-\cP_4$ and $\rho_4^\perp=\frac{\cQ_4}{8}$ as in Eq.~(\ref{eq: unextendible subspace}).

\begin{proposition}\label{prop: distillability}
    The state $\rho_4^\perp$ can be converted with nonzero probability, using stabiliser operations, into a single-qubit state in the distillable region of the five-qubit $T$-state protocol. In particular, $\rho_4^\perp$ is not bound magic.
\end{proposition}

\begin{proof}
    Consider the independent commuting Pauli operators $S_1=\one X\one Z$, $S_2=XXZ\one$ and $S_3=ZZ\one X$ and let
    \begin{align*}
        \Pi_{\cC}
        =\frac{1}{8}(\one+S_1)(\one+S_2)(\one+S_3)
        &=\frac{1}{8}(\one+ZZ\one X+XXZ\one-YYZX+\one X\one Z+ZY\one Y+X\one ZZ+YZZY)
    \end{align*}
    be the projector onto their common $(+1)$-eigenspace. Since the code has three independent stabilisers on four qubits, $\Pi_{\cC}$ has rank two. A choice of logical Pauli operators is $\overline X=XX\one\one$, $\overline Y=YXX\one$, $\overline Z=Z\one X\one$. These operators commute with $S_1,S_2,S_3$, anticommute pairwise, and satisfy $\overline Y=i\overline X\overline Z$.
    
    For a product of one-qubit Pauli eigenstates, the expectation of a Pauli string factorises, and a factor vanishes whenever the Pauli in the string differs from the local eigenbasis. Using this, one immediately computes
    \begin{align*}
        \bigl(\langle u_i|\Pi_{\cC}|u_i\rangle\bigr)_{i=1}^8
        &=\frac18(1,2,1,1,1,1,4,0)\; ,&
        \bigl(\langle u_i|\Pi_{\cC}\overline X|u_i\rangle\bigr)_{i=1}^8
        &=\frac18(1,-2,0,0,0,0,4,0)\; ,\\
        \bigl(\langle u_i|\Pi_{\cC}\overline Y|u_i\rangle\bigr)_{i=1}^8
        &=\frac18(0,0,0,-1,-1,0,0,0)\; ,&
        \bigl(\langle u_i|\Pi_{\cC}\overline Z|u_i\rangle\bigr)_{i=1}^8
        &=\frac18(0,0,-1,0,0,-1,0,0)\; .
    \end{align*}
    Consequently,
    \begin{align*}
        \tr[\cP_4\Pi_{\cC}]=\frac{11}{8}\; ,\qquad
        \tr[\cP_4\Pi_{\cC}\overline X]=\frac{3}{8}\; ,\qquad
        \tr[\cP_4\Pi_{\cC}\overline Y]=-\frac{1}{4}\; ,\qquad
        \tr[\cP_4\Pi_{\cC}\overline Z]=-\frac{1}{4}\; .
    \end{align*}
    Since $\tr[\Pi_{\cC}]=2$ and the logical Pauli operators are traceless on the code space, it follows that
    \begin{align*}
        \tr[\cQ_4\Pi_{\cC}]=\frac{5}{8}\; ,\qquad
        \tr[\cQ_4\Pi_{\cC}\overline X]=-\frac{3}{8}\; ,\qquad
        \tr[\cQ_4\Pi_{\cC}\overline Y]=\frac{1}{4}\; ,\qquad
        \tr[\cQ_4\Pi_{\cC}\overline Z]=\frac{1}{4}\; .
    \end{align*}
    Thus projection onto the code succeeds with probability $p_{\cC}=\tr[\Pi_{\cC}\rho_4^\perp]=\frac{5}{64}$. After postselection and Clifford decoding of the code space to one logical qubit, the output state is $\rho_L=\frac{1}{2}\left(\one-\frac{3}{5}X+\frac{2}{5}Y+\frac{2}{5}Z\right)$. Conjugation by $S^\dagger$, where $S=\operatorname{diag}(1,i)$, sends its Bloch vector to $(\frac{2}{5},\frac{3}{5},\frac{2}{5})$. Moreover, let $R$ be a Clifford unitary that cyclically permutes the Pauli operators and apply the Clifford twirl $\cT(\rho)=\frac{1}{3}\sum_{j=0}^2R^j\rho R^{-j}$. Then
    \begin{align*}
        \cT(S^\dagger\rho_LS)
        =\frac{1}{2}\left[\one+\frac{7}{15}(X+Y+Z)\right]\; .
    \end{align*}
    Writing a state on the $T$-axis as
    $\rho_T(r)=\frac{1}{2}\left[\one+\frac{r}{\sqrt3}(X+Y+Z)\right]$, the resulting polarisation is
    \begin{align*}
        r
        =\frac{7}{5\sqrt{3}}\; ,\qquad 
        r^2
        =\frac{49}{75}>\frac{3}{7}\; .
    \end{align*}
    The Bravyi--Kitaev distillation protocol \cite[Thm.~2]{BravyiKitaev2005} distils $T$-axis states for $r>\frac{3}{\sqrt{7}}$. Hence, a single copy of $\rho_4^\perp$ can, with probability $\frac{5}{64}$, be reduced by stabiliser operations to a state in that distillable region. Applying this reduction independently to many copies therefore distils arbitrarily pure $T$ states. It follows that $\rho_4^\perp$ is not bound magic.
\end{proof}

We note that the conclusion is inherited by all `saturated' ancillary lifts of $\cU_4$. More precisely, if $\cB_A$ is a complete stabiliser basis of an ancillary system and $\cU_4^\uparrow=\{|u_i\rangle\otimes|b\rangle\mid 1\leq i\leq8,\ |b\rangle\in\cB_A\}$, then $\cQ_{\cU_4^\uparrow}=\cQ_4\otimes\one_A$ and $\rho_{\cU_4^\uparrow}^\perp=\rho_4^\perp\otimes\frac{\one_A}{D_A}$. Discarding the ancilla recovers $\rho_4^\perp$, hence, every such lifted complement is also magic-distillable.\\

\textbf{Distinguishability by stabiliser operations.} Finally, we show that stabiliser unextendibility gives no uniform quantitative obstruction to discrimination by stabiliser operations (SO) \cite{VeitchEtAl2014,HowardCampbell2017}. Our discrimination protocol below falls within the class of stabiliser operations, and consequently, the same lower bound on the discrimination success probability therefore holds for any operational class containing SO (e.g. the generally larger class of completely stabiliser-preserving channels \cite{HeimendahlHeinrichGross2022}).

\begin{proposition}\label{prop: approximate SO distinguishability}
    Let $\cU=\{|\psi_i\rangle\}_{i=1}^k$ be an $n$-qudit USB. For every $\varepsilon>0$, there exists a USB $\cV$ in a larger system such that the uniform-prior discrimination success probability under stabiliser operations satisfies
    \begin{align*}
        p_{\mathrm{succ}}^{\mathrm{SO}}(\cV)
        >1-\varepsilon\; .
    \end{align*}
\end{proposition}

\begin{proof}
    Let $\cB={|b_j\rangle}_{j=1}^D$ be the computational basis. Introduce an $m$-qudit control register, let $A=q^m$, choose a computational-basis label $t_\star\in\mathbb F_q^m$, and define the set,
    \begin{align*}
        \cV_m
        =\{|t\rangle\otimes|b_j\rangle\mid t\neq t_\star,\ 1\leq j\leq D\}
        \cup\{|t_\star\rangle\otimes|\psi_i\rangle\mid 1\leq i\leq k\}\; ,
    \end{align*}
    of $(A-1)D+k<AD$ pairwise orthogonal stabiliser states. Unextendibility follows as in Lm.~\ref{lm: stabiliser uncompletability}: suppose that a pure stabiliser state $|\Phi\rangle$ were orthogonal to every element of $\cV_m$. For each computational-basis label $t$, let
    \begin{align*}
        |\phi_t\rangle
        =(\langle t|\otimes\one)|\Phi\rangle\; .
    \end{align*}
    Every nonzero $|\phi_t\rangle$ is proportional to a pure stabiliser state. If $t\neq t_\star$, then $|\phi_t\rangle$ is orthogonal to the complete basis $\cB$, and hence $|\phi_t\rangle=0$. The remaining branch $|\phi_{t_\star}\rangle$ is orthogonal to every member of $\cU$, and therefore also vanishes by unextendibility of $\cU$. Together, this implies $|\Phi\rangle=0$, hence, $\cV_m$ is a USB.
    
    Next, we construct an explicit stabiliser discrimination protocol. Measure the control register in the computational basis. If the outcome is $t\neq t_\star$, measure the second system in the basis $\cB$, thus identifying the input state perfectly. If the outcome is $t_\star$, apply any discrimination protocol for $\cU$ with success probability $p_{\mathrm{succ}}^{\mathrm{SO}}(\cU)$. Consequently,
    \begin{align*}
        p_{\mathrm{succ}}^{\mathrm{SO}}(\cV_m)
        \geq\frac{(A-1)D+kp_{\mathrm{succ}}^{\mathrm{SO}}(\cU)}{(A-1)D+k}\; ,
    \end{align*}
    for the overall success probability of the protocol. Since guessing succeeds with probability $\frac{1}{k}$, this implies
    \begin{align*}
        1-p_{\mathrm{succ}}^{\mathrm{SO}}(\cV_m)
        \leq\left(1-\frac{(A-1)D+1}{(A-1)D+k}\right)
        =\frac{k-1}{(A-1)D+k}\longrightarrow0\; 
    \end{align*}
    as $m\rightarrow\infty$. Choosing $m$ sufficiently large thus proves the claim.
\end{proof}

UPBs universally obstruct perfect LOCC discrimination \cite{BennettEtAl1999}. By contrast, Prop.~\ref{prop: approximate SO distinguishability} shows that USB unextendibility alone provides no uniform quantitative obstruction to stabiliser discrimination. As a consequence, USBs also do not have a dimension-independent discrimination or data-hiding gap. This contrasts with Kwon's recent $3$-qubit example, which cannot be perfectly distinguished by stabiliser operations \cite{Kwon2025}.

\end{document}